\documentclass{article}
\usepackage[margin=2.5cm]{geometry}
\usepackage{authblk}

\usepackage{graphicx}%
\usepackage{multirow}%
\usepackage{amsmath,amssymb,amsfonts}%
\usepackage{amsthm}%
\usepackage{mathrsfs}%
\usepackage[title]{appendix}%
\usepackage{xcolor}%

\usepackage{textcomp}%
\usepackage{manyfoot}%
\usepackage{booktabs}%
\usepackage{algorithm}%
\usepackage{algorithmicx}%
\usepackage{algpseudocode}%
\usepackage{listings}%

\usepackage[T1]{fontenc}
\usepackage[american]{babel}
\usepackage{natbib}
\usepackage{amssymb}
\usepackage{graphicx}
\usepackage{hyperref}
\usepackage{soul}
\usepackage{bbm}
\usepackage{lscape}
\usepackage{array}
\newcolumntype{L}{>{\centering\arraybackslash}m{3.6cm}}

\usepackage{amsthm}
\newtheorem*{remark}{Remark}

\newtheorem{definition}{Definition}
\newtheorem{proposition}{Proposition}

\usepackage[normalem]{ulem} 
\usepackage{enumitem} 

\newcommand{\bfone}{\boldsymbol{1}} 

\newcommand{\bfalpha}{\boldsymbol{\alpha}}
\newcommand{\bftheta}{\boldsymbol{\theta}}
\newcommand{\bfrho}{\boldsymbol{\rho}}

\newcommand{\sa}{\text{SA}}
\newcommand{\sr}{\text{SR}}
\newcommand{\ssr}{\text{SSR}}
\newcommand{\OR}{\text{OR}}
\newcommand{\rr}{\text{RR}}
\newcommand{\starr}{\text{STARR}}
\newcommand{\esgd}{\text{ESG-}}

\newcommand{\esg}{{ESG}} 
\newcommand{\cash}{\text{CASH}}
\newcommand{\esgm}{\text{esg}}

\newcommand{\E}{\mathbb{E}}

\newcommand{\avar}{\text{AVaR}}

\newcommand{\bfzero}{\boldsymbol{0}}

\newcommand{\esgavar}{\text{ESG-} \avar_{\lambda,\tau}}

\newcommand{\ftr}{\text{FTR}}
\newcommand{\rft}{\text{RF}_T^r}
\newcommand{\rftesg}{\text{RF}_T^{\esg}}

\newcommand{\Eb}{\mathbb{E}}

\newcommand{\Rb}{\mathbb{R}}
\newcommand{\Zc}{\mathcal{Z}}
\newcommand{\Xc}{\mathcal{X}}
\newcommand{\Ac}{\mathcal{A}}
\newcommand{\Lc}{\mathcal{L}}
\newcommand{\Fc}{\mathcal{F}}
\newcommand{\one}{\mathbbm{1}}

\newcommand{\cm}{\checkmark}

\usepackage{setspace} 

\begin{document}

\onehalfspacing

\title{An Axiomatic Risk-Reward Framework for Sustainable Investing}

\author[a]{Gabriele Torri\thanks{Corresponding author: \textit{gabriele.torri@unibg.it}}}
\author[a]{Rosella Giacometti}
\author[b]{Darinka Dentcheva}
\author[c]{Svetlozar T. Rachev}
\author[c]{W. Brent Lindquist}

\affil[a]{Department of Management, University of Bergamo, Bergamo 24127, Italy}
\affil[b]{School of Engineering and Science, Stevens Institute of Technology, Hoboken, NJ 07030, USA}
\affil[c]{Department of Mathematics and Statistics, Texas Tech University, Lubbock, TX 79409-1042, USA}

\date{\today}

\maketitle

\begin{abstract}
Continued interest in sustainable investing calls for an axiomatic approach to measures of risk and reward that focus not only on financial returns, but also on measures of environmental and social sustainability, i.e. environmental, social, and governance (ESG) scores. We propose axiomatic definitions for \textit{ESG-coherent risk measures} and \textit{ESG reward--risk ratios} based on functions of bivariate random variables that are applied to financial returns and real-time ESG scores, extending the traditional univariate measures to the ESG case. We provide examples, discuss the dual representation, and present an empirical analysis in which the ESG-coherent risk measures and ESG reward--risk ratios are used to rank stocks.

\textbf{Keywords:} risk measures, multivariate risk, ESG, sustainable investing
\end{abstract}




\section{Introduction}

ESG investing refers to the integration of environmental, social, and governance considerations into the asset allocation process. It has been one of the most significant trends in the asset management industry, due to continued focus on sustainability and to the growth of information related to non-financial impacts.

ESG investing encompasses a broad array of approaches, and its market practices are very heterogeneous, with different terminologies, definitions, and strategies.
These practices vary due to the cultural and ideological diversity of investors \citep{sandberg2009heterogeneity,widyawati2020systematic}. 
According to \cite{amel2018and}, the most significant motivation for incorporating ESG factors is related to financial performance, as sustainability factors are perceived as relevant to investment returns. That is, investors believe that ESG data can be used to identify potential risks and opportunities, and that such information is not yet fully incorporated into market prices. Hence, ESG information should help investors to control risk better and improve their financial performance.
In line with \cite{schanzenbach2020reconciling} we employ the expression \textit{risk-return ESG} to refer to investment strategies that use ESG factors to improve returns while lessening risk. Academic evidence on the role of ESG in enhancing performance is inconclusive. The meta-analysis conducted by \cite{revelli2015financial} shows that sustainable and responsible investing (SRI) is neither a weakness nor a strength compared with conventional investing. In contrast, \cite{friede2015esg} combining the findings of approximately 2200 studies from 1970 to 2014 find that large majority of of the considered literature reports a positive relationship between ESG and corporate financial performance. More recently, \cite{hornuf2024performance} found that responsible investment neither outperforms nor underperforms the market portfolio.\footnote{We note that, if ESG information is already incorporated in market prices, any strategy that restricts investment based on ESG criteria would result in sub-optimal allocation in terms of monetary performance.
	Moreover we underline that risk-return ESG strategies are not necessarily more sustainable than traditional ones. Indeed, an investor may implement a contrarian-ESG approach, investing in less sustainable assets if they are expected to outperform the market.}

A second motivation that guides ESG strategies is the desire to improve the sustainable profile of the portfolios for ethical reasons or to improve investors' green image \citep{amel2018and}.\footnote{Almel-Zadeh and Serafim reported additional motivations that push asset managers to propose ESG products, including growing client demand, the effectiveness of ESG investing in bringing about change in companies, and formal client mandates. For our purposes, these drivers can be attributed to either material or ethical motivations.} 
Sustainability then becomes part of the investment goals, alongside monetary performance and the riskiness of the position.
\cite{schanzenbach2020reconciling} refer to investment strategies that incorporate ESG screenings for moral or ethical reasons as \textit{collateral benefit ESG}, as they aim to provide benefits to a third party, rather than to improve risk-adjusted returns.
We refer to investors who include sustainability considerations for ethical reasons as \textit{ESG-oriented investors}, to distinguish them from \textit{risk-return investors} who care exclusively about the financial performance of a position.

From a theoretical point of view, risk-return ESG does not pose any specific issue, as ESG is treated as any other information (e.g. balance sheet data, macroeconomic indicators, sentiment analysis, etc.) and is integrated in the investment process without affecting the main goal of the investor: improving the risk-adjusted performance.
In contrast, the ethical motivations of ESG-oriented investors challenge many of the traditional assumptions adopted in finance theory, as the inclusion of other factors in the investment process (e.g. reduction of CO$_2$ emissions, respect of human right, exclusion of controversial sectors, etc.) are not necessarily motivated by the improvement of financial performance.

To illustrate how the ethical motivation breaks common investing principles, consider two stocks with the same reward--risk profile, but with different ESG scores. They may be equivalent  for a risk-return investor but not for an ESG-oriented investor, as one of the two companies may be more environmentally sustainable or may adopt stricter human right policies.
An ESG-oriented investor may decide to deliberately worsen their risk-adjusted performance to comply with non-negotiable principles, for instance by excluding certain sectors or companies, thus reducing the diversification of their portfolio.\footnote{It is common practice for sustainable investment funds to exclude tobacco, weapons, fossil fuels, and other controversial sectors. The divestment campaign aimed at South Africa's apartheid regime in the 1980s was one of the key turning points in the history of responsible investing \citep{schanzenbach2020reconciling}.}

The problem of designing optimal portfolio strategies for ESG-oriented investors that jointly account for financial performance and sustainability considerations has been studied in the portfolio optimization literature. This line of research typically extends the classical reward-risk framework by introducing a third dimension -- sustainability -- proxied by ESG scores. The scores are incorporated into the portfolio selection problem either as explicit constraints or as additional terms in the objective function. Such an approach has been pursued by, among others, \cite{utz2015tri}, \cite{gasser2017markowitz} and \cite{cesarone2022does} who studied the efficient frontier of non-dominated portfolios in the reward-risk-ESG space, extending the traditional optimal portfolio literature that started with \cite{markowitz1952portfolio}. These works aim to find portfolios with the highest expected return for a given level of risk, where risk is measured using the portfolio variance, a coherent risk measure such as the average value at risk (AVaR)\footnote{The AVaR is also known in the literature as Conditional Value at Risk (CVaR) or Expected Shortfall (ES)} \citep{rockafellar2000optimization}, or an asymmetric deviation measure \citep{giacometti2021tail}. \cite{pedersen2021responsible} further contributes to the literature by studying the effect on the market equilibrium: they build a framework with three classes of investors (\textit{ESG unaware}, that do not consider ESG scores; \textit{ESG-aware} that use sustainability scores to update their views on risk and expected return; and \textit{ESG motivated}, that have preferences for high ESG) and they consider an ESG-efficient frontier by considering the Sharpe ratio and the ESG score. The presence of sustainability-oriented investors in the market and they effect on asset prices has been studied also by \cite{heinkel2001effect} that investigated the effects of exclusionary ethical investments on corporate behaviour, finding that if the percentage of ethical investors is sufficiently high, the cost of capital for polluting companies may increase. More recently, \cite{pastor2021sustainable} built a two-factor model to study the market equilibrium in the presence of ESG investors, finding that, in equilibrium, sustainable assets have lower expected returns not only due to the fact that they can be used to hedge climate risk, but also because of the preference of ESG investors towards these assets.

One drawback shared by the modeling frameworks discussed so far is that they fail to take into account the stochasticity of ESG scores: using either the expected values of the scores or assuming that ESG scores can be treated as constants. Here, we aim to address this limitation and, more generally, to introduce a framework to jointly consider the financial and ESG dimensions.

In this work we target explicitly ESG-oriented investors, proposing measures that account for the joint distribution of financial returns and ESG scores. In particular, we propose axiomatic classes of measures that extend the concepts of coherent risk measures \citep{artzner1999coherent}, reward measures \citep{rachev2008desirable}, and reward-risk ratios \citep{cheridito2013reward}. We underline that our framework treats ESG as an explicit argument of investor preferences rather than as a return-predictive signal. Still, the approach can be integrated in asset allocation strategies that jointly incorporate ESG-based preferences and ESG-driven information for financial performance.

 Our axiomatic approach is based on the idea that risk comes from two drivers: monetary performance (i.e., the financial returns of a position) and sustainability (represented by the ESG score of the company).\footnote{For convenience, in the rest of this paper, we refer to these two dimensions as the \textit{monetary} and \textit{ESG} components of risk.} These quantities are random and not necessarily independent, and we can represent them as a bivariate random variable. Hence, it is natural to refer to the rich literature on multivariate risk measures \citep{jouini2004vector,hamel2009duality,wei2014coherent,ekeland2012comonotonic}. Such measures have been developed to study portfolios of non-perfectly fungible assets (e.g., assets valued in multiple currencies) or assets that are difficult to price. We propose using a bivariate risk measure to deal with a single asset that can be evaluated over two dimensions: the monetary returns and the sustainability (represented by ESG scores). We then define \textit{ESG-coherent risk measures} as an extension of the \textit{coherent risk measures} introduced by \cite{artzner1999coherent}. 

Since different investors may have different attitudes towards ESG, the proposed measures are parametrized using a value $\lambda \in [0,1]$ to explicitly take into account the subjective trade-off between sustainability and financial performance: when $\lambda = 0$ an investor cares exclusively about financial risk, when $\lambda = 1$ the investor cares only about ESG. In addition to ESG-coherent risk measures, we define \textit{ESG-coherent reward measures} and \textit{ESG reward--risk ratios}.

The proposed ESG-coherent measures can be directly embedded in optimal portfolio frameworks (in the objective function or as constraints), with minimal modifications compared to the univariate case. Unlike ESG constraints or screening rules, ESG enters through the risk functional itself, preserving convexity and dual representations. Portfolio choice is therefore formulated as a well-defined optimization problem based on the joint distribution of financial returns and ESG outcomes.

From an applied perspective, the analysis is motivated by market participants’ increasing focus on sustainability and by the emergence of high-frequency ESG measures, such as the \textit{Truvalue scores} issued by FactSet, as opposed to the largely annual disclosure cycle of conventional ESG data.

Section \ref{sec:measure_ESG} discusses our interpretation of ESG scores and outlines how we propose to use them. Section \ref{sec:risk_measures} introduces ESG-coherent risk measures and provides examples, while Section \ref{sec:RRR} introduces ESG reward--risk ratios. Section \ref{sec:empirical} presents an empirical application using the equity returns and real-time ESG data of the constituents of the Dow-Jones Industrial Average (DJIA) financial index. Section \ref{policy} discusses policies and practical applications. Section \ref{sec:conclusions} highlights some conclusions.

\section{Measuring ESG and financial performance}\label{sec:measure_ESG}

The first step in defining ESG risk measures is to clarify how we characterize and measure sustainability. Following an approach common in the literature, we use the ESG scores of a company as a proxy for sustainability \citep[see, e.g.,][]{pedersen2021responsible,utz2015tri,cesarone2022does}. Here, we clarify how to treat such a variable from both practical and theoretical points of view, establishing the basis for measuring ESG scores in a way that allows us to consistently use them alongside monetary values for risk measurement. 
We emphasize that, at this stage, we keep the discussion on an abstract level, without reference to specific data providers or methodologies. This allows us not to be limited by the current state-of-the-art market ESG scores, which are still far from being standardized and comparable across data providers \citep[see, e.g.,][]{billio2021inside, berg2022aggregate,lauria2026mean}.

Several approaches appear in the literature for modeling the random variable $r_T$ used to compute risk. \cite{artzner1999coherent} define coherent risk measures for a random variable that represents net worth by following the principle that ``bygones are bygones,'' meaning that future net worth is the only thing that matters. In an alternative approach, the random variable used to compute monetary risk represents the return of a financial position. This latter approach, which is often used in practical applications for the measurement of the risk of equity portfolios,  can be interpreted as the payoff of a unitary investment in a risky position, and it introduces some differences in the interpretation of the axioms \cite[see][Chapter 6]{rachev2011risk}.\footnote{Alternatively, some authors prefer to define risk measures computed on a variable that represents losses, thus assuming that lower values of the variable are preferred by an investor. Such an assumption does not significantly alter the analysis: it simply changes some signs in the definition of risk (see, e.g., \citealp{rockafellar2013fundamental}).}

Similarly, we need to establish the quantity measured by the random variable $\esg_T$ and the sign convention used. In particular, one must consider whether it is a \textit{stock variable} (measured at a specific point in time) or a \textit{flow variable} (measured over an interval of time, as some sort of \textit{sustainability return}). Based on how they are computed, we argue that ESG scores belong to the second category: indeed, they represent the current level of sustainability of the production and commercial practices of a company, which directly affects the  impact of that company on the world.\footnote{
	The scores are typically computed as functions of several indicators related to the production methods,
	the supply chain management, the industry in which the company operates,
	the transparency of its governance, the presence of specific policies on human rights violations, etc.} 
Unlike static ESG ratings, our interpretation treats sustainability as an accumulated flow of externalities (positive or negative) that are generated over time. In this sense, the ESG score of a company is related to the rate at which it accumulates non-monetary ``satisfaction'' for the investor. The total non-monetary satisfaction for an investor is proportional to the holding time of the investment.

Our approach is thus to consider a stochastic process $\esgm_t$, that describes the instantaneous \textit{sustainability flow}.\footnote{This quantity, although related to externalities, does not necessarily need to be expressed in monetary terms, as the value of this component is different for each investor. A monetary market price of sustainability may exist, but each investor may assign a subjective value to it.} 
For a given time horizon, the satisfaction of the investor depends on the amount of sustainability accumulated over time, which, for the interval of time from $0$ to $T$, is defined by
\begin{equation}
\esg_{T}:= \int_{0}^{T} \esgm_t dt.
\end{equation}

The variable $\esg_T$ aggregates sustainability performance over the holding period. We can then interpret the ESG scores assigned by market providers with periodicity $T$ (e.g. a day or a year), as (rescaled) realizations of the random variables $\esg_{T}$. 

Finally, for each asset we define a bivariate random variable $X_{T}$ used to compute ESG risk and reward measures, as: 

\begin{equation}
	X_T = \begin{bmatrix}
		r_T\\
		\esg_T
	\end{bmatrix},
\end{equation}

where $\esg_{T}$ is the sustainability dimension, and $r_{T}$ is the monetary dimension, computed as the cumulative log-return over a static-horizon of length $T$.

Formally, our framework works under the assumption that the bivariate random vector $X_T$ is measurable on the same probability space $(\Omega,\mathcal{F},{P})$. We acknowledge that in practical settings the ESG score process may depend on external information, or that sustainability related information may affect the dynamic of financial returns. We point out that the theoretical axiomatic setup does not impose structural assumptions beyond measurability and joint distribution of $X_T$.

The stochasticity of the process $\esgm_t$ can be traced to two sources: the first is the uncertainty in the measurement of its value, and the second is the uncertainty in the evolution of the sustainability policies within the company, which is driven by the choices made by the management, by market conditions, and by the regulatory framework. We also expect a correlation structure between $\esg_T$ and financial returns $r_T$ due to the presence of common driving factors related, for instance, to sector-wide or country-wide dynamics. 

In the rest of the work we use the following scaling convention for the random variable $\esg_T$: assuming that the ESG scores are bounded, we rescale the values of daily ESG scores such that $\esg_{T} \in [-1/c,1/c]$, where $T$ is equal to 1 trading day, and $c=252$. As discussed later, when applied to the Factset Truvalue ESG scores this standardization leads to  the ``typical magnitude'' of $\esg_{T}$ for a daily time period to be comparable to the ``typical magnitude'' of the financial cumulative returns measured over the same time period.

The periodic return $r_T$ may be viewed as the payoff of a unit investment in a risky position. Consequently, $\esg_T$ is taken to refer to the same unit investment, so that the variable is proportional to the initial size of the position. With respect to sign conventions, it is natural to define ESG scores so that higher values are preferable, in line with the standard preference for higher returns. These conventions are adopted for comparability and do not restrict the generality of the axiomatic results.

Concerning empirical applications, we acknowledge that the quality and standardization of ESG scores across data providers is a major issue. At the time of writing, most of the available ESG scores are largely based on balance sheet items and self-reported information, and are updated annually. This makes it difficult to obtain high frequency and forward-looking information regarding the sustainability of a company. One notable exception is the \textit{Truvalue scores} framework introduced by Factset, that are computed with daily frequency. Such scores will be used in the empirical analysis in Section \ref{sec:empirical}. With the growing use of real-time data, the relevance of properly modelling the time-series dynamics of ESG, and the integration in quantitative risk and portfolio management will grow, making it relevant to develop a solid theoretical framework. Thus, further research should focus on the modeling and estimation of $\esgm_t$, considering both the time series evolution of reported ESG scores, and possibly the dispersion of ratings from different providers.

\section{ESG-coherent risk measures}\label{sec:risk_measures}

We recall the axioms that define a coherent risk measure before introducing the bivariate framework \citep{artzner1999coherent}. These axioms make it possible to identify measures with desirable properties, assisting both investors and regulators.
Consider a convex set $\mathcal{X}\subseteq\mathcal{L}_p(\Omega,\mathcal{F},P)$ of real-valued random variables $r_T$ which are defined on a probability space $(\Omega,\mathcal{F},P)$, have finite $p$-moments ($p\geq 1$), and are indistinguishable up to a set of $P$-measure zero. We assume that the random variables represent the cumulative return over time T or the payoff at time T of an asset. The functional $\rho(r_T) : \mathcal{X} \rightarrow \mathbb{R}\cup \{+\infty\}$ is a coherent risk measure if it satisfies the following properties:
\begin{itemize}[align=left]
	\item[(SUB)] sub-additivity: if $r_{1,T},r_{2,T} \in \mathcal{X}$, then $\rho(r_{1,T}+r_{2,T} ) \leq \rho(r_{1,T}) + \rho(r_{2,T} )$;
	\item [(PH)] positive homogeneity: if $\alpha\in \Rb_+$ and $r_T \in \mathcal{X}$, then $\rho(\alpha r_T) = \alpha \rho(r_T)$;
	\item [(TI)] translation invariance: if $a$ is deterministic, $\rho(r_T+a)=\rho(r_T)-a$;
	\item[(MO)] monotonicity: if $r_{1,T},r_{2,T} \in \mathcal{X}$ and $r_{1,T}\leq r_{2,T}$ a.s., then $\rho(r_{1,T}) \geq \rho(r_{2,T})$.
\end{itemize}
Examples of coherent risk measures are AVaR and expectile; but VaR and standard deviation are not coherent risk measures.\footnote{
	Alternative axiomizations of risk measures, such as convex risk measures \citep{follmer2002convex}
	and regular risk measures \citep{rockafellar2013fundamental} have been proposed since the work of \citep{artzner1999coherent}.}

The definition of a coherent risk measure is based upon a univariate random variable, and on the idea that the risk is a function only of the returns of the asset (consistent with the assumption that an investor is only interested in the monetary outcome of their position). In some contexts, the monetary outcome does not fully characterize the risk. Examples include: portfolios in two countries having floating exchange rates whose payoffs in the different currencies are not perfectly substitutable (the Siegel paradox, see \citealp{black1989universal}); the fact that various maturities for interest rate products are not perfect substitutes (failure of the pure expectation hypothesis); and cases in which it is difficult to attribute a monetary equivalent to various dimensions of risk, such as environmental or health risks. The same principle can be used for the ESG-oriented investor, whose risk depends on both the monetary return of a financial position and its sustainability, represented in our analysis by its ESG score, that is not priced in monetary units. Such an approach allows an investor to deal with trade-offs between different goals that, although they are very different in nature, have to be taken into account in the investment process. 

\subsection{ESG-coherent risk measures --- axiomatic definition} \label{sec:axioms}

We now extend the classical axioms of coherent risk to a bivariate setting, where the two components represent monetary and sustainability dimensions of the same position.\footnote{
	As discussed in Section \ref{sec:measure_ESG}, we use periodic returns for the monetary component and ESG scores (which represent the accumulated sustainability of the investment)
	for the sustainability component.
	The definition could be extended to alternative specifications that can be represented using two random variables with a joint distribution such that an investor has a preference for both higher $r$ and higher $\esg$ values
	(e.g., using the final net worth and the accumulated sustainability multiplied by the initial value of the position).}
Consider a convex set $ \mathcal{X}_2$ of random vectors $X_T = [ r_T, \; ESG_T ]'$, defined on a probability space $(\Omega,\mathcal{F},P)$ with values in $\mathbb{R}^2$. We use the short-hand notation $\mathcal{X}_2=\mathcal{L}_p(\Omega,\mathcal{F},P;\Rb^2)$ for the space of random vectors with two components and finite $p$-th moments, which are indistinguishable on sets of $P$-measure zero. Here $p\in[1,\infty]$. Since individual investors may have different attitude towards ESG, to highlight the trade-off between the monetary and sustainability risk components, we use the parameter $\lambda \in [0,1]$, a scalar which represents an investor preference for the relative weighing between the monetary and ESG components of risk. Such parameter is investor-specific but fixed at evaluation time.

An ESG risk measure is then a functional of the form $\bfrho_{\lambda}(X): \mathcal{X}_2 \rightarrow \mathbb{R} \cup \{+ \infty\}$. As in the univariate case, it is possible to axiomatically characterize a set of measures that have desirable properties. These axioms are:

\begin{itemize}[align=left]
	\item[(SUB-M)] sub-additivity: if $X_{1,T},X_{2,T} \in \mathcal{X}_2$, then $\bfrho_\lambda (X_{1,T}+X_{2,T}) \leq \bfrho_\lambda(X_{1,T}) + \bfrho_\lambda(X_{2,T})$;
	\item [(PH-M)] positive homogeneity: if $\beta \in \Rb_+$ and $X_{T} \in \mathcal{X}_2$, then $\bfrho_\lambda(\beta X_{T}) = \beta \bfrho_\lambda(X_{T})$;
	\item[(MO-M)] monotonicity: if $X_{1,T},X_{2,T} \in \mathcal{X}_2$ and ($r_{1,T}\leq r_{2,T} \wedge \esg_{1,T}\leq \esg_{2,T}$) a.s., then $\bfrho_\lambda(X_{1,T}) \geq \bfrho_\lambda(X_{2,T})$;
	\item [(LH-M)] lambda homogeneity: if $a = [a_1, \; a_2]' \in \mathbb{R}^2$ is deterministic, then $\bfrho_\lambda(a) = -\big((1-\lambda)a_1 + \lambda a_2\big)$.
\end{itemize}

The first three axioms directly extend the univariate case, while LH-M axiom specifies the risk of deterministic ESG-monetary bundles and encodes investor preferences. We then provide the following definition:

\begin{definition}[ESG-coherent risk measure]
	Consider a probability space $(\Omega, \mathcal{F}, P)$, a parameter $\lambda \in [0,1]$, and $X_T = [ r_T, \; ESG_T ]'$ belonging to a set of bivariate random variables $\mathcal{X}_2$ where $r_T$ measures the cumulative returns of a position or portfolio over a period $T$, and $\esg_T$ measures the cumulative sustainability flow. We define an ESG-coherent risk measure as any functional $\bfrho_\lambda(X_T): \mathcal{X}_2 \rightarrow \mathbb{R} \cup \{+ \infty\}$ that satisfies the four axioms SUB-M through LH-M. 
\end{definition}

It follows from the axioms that any ESG-coherent risk measure satisfies the following property (TI-M), which generalizes translation invariance to the bivariate case.

\begin{proposition} \label{prop:TI}
	An ESG-coherent risk measure satisfies the translation invariance property TI-M:
	\begin{itemize}
		\item [(TI-M)] translation invariance: if $a = [a_1, \; a_2]' \in \mathbb{R}^2$ is deterministic, $\bfrho_{\lambda}(X_{T}+a)=\bfrho_{\lambda}(X_{T}) -\big((1-\lambda)a_1 + \lambda a_2\big)$;
	\end{itemize}	
\end{proposition}

\begin{proof}
	for SUB-M we know that, for $a = [a_1, \; a_2]' \in \mathbb{R}^2$ 
	\begin{equation}
		\bfrho_{\lambda}(X_{T}+a) \leq \bfrho_{\lambda}(X_{T}) + \bfrho_{\lambda}(a). \label{eq:proof_TI1}
	\end{equation}
	
	Using SUB-M and noting that due to LH-M we have $\bfrho_{\lambda}(-a) = -\bfrho_{\lambda}(a)$, we show that $$\bfrho_{\lambda}(X_{T}) = \bfrho_{\lambda}(X_{T}+a-a) \leq \bfrho_{\lambda}(X_{T}+a) - \bfrho_{\lambda}(a)$$ that implies 
	\begin{equation} 
		 \bfrho_{\lambda}(X_{T}+a) \geq \bfrho_{\lambda}(X_{T}) + \bfrho_{\lambda}(a). \label{eq:proof_TI2}
	\end{equation}
	 Considering axiom LH-M and equations (\ref{eq:proof_TI2}) and (\ref{eq:proof_TI2}), we have: $$ \bfrho_{\lambda}(X_{T}+a) = \bfrho_{\lambda}(X_{T}) -\big((1-\lambda)a_1 + \lambda a_2\big).$$
\end{proof}

We note the following.
\begin{itemize}[align=left]
\item Axioms SUB-M and PH-M are straightforward multivariate generalization of the axioms SUB and PH in the univariate definition of coherent risk measures.
\item Axiom MO-M generalizes MO. Specifically, we impose a monotonicity condition for which zeroth-order stochastic dominance implies the a.s. ordering of risks. Alternative approaches may be based on first- and second-order stochastic dominance; we leave such an analysis for future studies and maintain the most general (weakest) condition. We emphasize that the ordering induced by the rule in MO-M is partial.
\item The last axiom LM-H defines the value of the risk measure for a constant quantity and introduces the parameter $\lambda$ that accounts for the subjective preference of each investor between financial and ESG dimensions. 
\item The axiom LM-H enables the characterization of the risk of an \textit{ESG safe asset} (SA), i.e. a position having constant values for both monetary and ESG components. This specification ensures that $\bfrho_{\lambda}(\bfone)  = -1 , \; \forall \lambda \in [0,1]$, where $\bfone$ is the vector $[1,1]'$.
\item Property TI-M extends the univariate translation invariance axiom TI.
\end{itemize}

This set of axioms is related to the work of \cite{ruschendorf2006law}, \cite{wei2014coherent}, and \cite{chen2020multivariate}, which define convex risk for portfolio vectors using scalar-valued functions. A key difference from these papers, is the introduction of the parameter $\lambda$ in the axiom LM-H.

The class of ESG-coherent risk measures extends coherent measures to the multivariate setting and provides a way to control the trade-off between the two sources of risk. This trade-off depends on the preferences of the individual investor expressed by $\lambda$. Axioms SUB-M and PH-M guarantee that an ESG-coherent risk measure is convex. This allows an investor to diversify not only by creating portfolios of multiple assets, but also to diversify between monetary risk and ESG risk, as highlighted by the following remarks.
\begin{remark}
	Given $X_T = [r_T, \; ESG_T ]' \in \mathcal{X}_2$ and an ESG risk measure $\bfrho_\lambda(X_T)$, for an investor with a given $\lambda$ the \textit{pure monetary risk} and \textit{pure ESG risk} are defined by $\bfrho_\lambda([r_T, \;0]')$ and $\bfrho_\lambda([0, \;\esg_T]')$, respectively.
\end{remark}

\begin{remark}[Diversification between monetary and ESG risk]
	If $\bfrho_{\lambda}(X_T)$ is an ESG-coherent risk measure, from SUB-M we observe that
	\begin{equation}
		\bfrho_{\lambda}(X_T)\leq \bfrho_\lambda([r_T, \;0]') + \bfrho_\lambda([0, \;\esg_T]').
	\end{equation}
That is, the risk of a position is always less than or equal to the sum of the pure ESG risk and the pure monetary risk (the investor diversifies between the ESG risk and monetary risk).
\end{remark}

Further insights can be gathered by considering the acceptance set characterization of univariate risk measures: as discussed in \cite{artzner1999coherent}, in the univariate case risk can be interpreted as the minimum amount of cash that has to be added to a position to make it acceptable (i.e. to have risk smaller or equal than zero). Extending to a multivariate setting, we no longer have a single safe asset (cash), but rather several instruments with deterministic payoffs. More specifically, in the context of ESG investing, one can postulate the presence of multiple such instruments, each characterized by a different combination of return and ESG attributes. In this case, it is possible to make a risky position acceptable by adding specific amounts of the available safe assets. Indeed, an alternative set-valued multivariate extension discussed in the literature defines risk as the set of all the deterministic capital positions that added to a risky position make it acceptable. The advantages of the set-valued approach is to provide a more complete assessment of risk in relation to the multiple drivers, and to provide a general mathematical formulation, but at the cost of greater complexity and the inability to directly rank positions. We limit our analysis to the scalar-valued bivariate risk measures illustrated above, and we refer to \cite{jouini2004vector,hamel2009duality,hamel2013set,feinstein2013time} for a broader discussion on set-valued multivariate risk. The practical implication for ESG-oriented investors of the presence of multiple safe assets in the market, are further discussed in Section \ref{sec:risk_free}.

We note that the definition of an ESG-coherent risk measure remains agnostic concerning the measurement of either the financial performance or the ESG score; the former can be measured in terms of the final wealth, profit and loss, or periodic returns \citep{artzner1999coherent}, and the latter can be computed according to multiple methodologies and aggregated over time following several approaches. The only requirement is that the investor must have a preference for both higher financial gain
and higher ESG scores (hence, the monetary part must be expressed in terms such that gains are positive).

In an analogous manner, we can define ESG-coherent reward measures that extend the work of \cite{rachev2008desirable} (see Appendix \ref{sec:reward}).

\subsubsection{Dual representation}
It is well known that coherent risk measures have a dual representation $-$ the supremum of a certain expected value over a risk envelope \citep{ruszczynski2006optimization,ang2018dual,dentcheva2024risk}. For ESG-coherent risk measures, the dual representation is introduced in Proposition \ref{prop:dr}.

\begin{proposition}\label{prop:dr}
	Given $X_T = [r_T, \; ESG_T ]' \in \mathcal{X}_2$ and an ESG-coherent risk measure $\bfrho_\lambda(X_T)$ that satisfies axioms SUB-M through LH-M, the dual representation of the risk measure is 
	\begin{align}\label{eq:dual}
	\bfrho_\lambda (X_T)& =\sup_{\zeta\in \Ac_{\bfrho_\lambda}} \left\{- \int_\Omega \left[ \zeta_1(\omega) r_T (\omega)+ \zeta_2(\omega) \esg_T (\omega)\right] P(d\omega) \right\},
	\end{align}
	where $\Ac_{\bfrho_\lambda}$ contains non-negative functions $(\zeta_1,\zeta_2)\in \Lc_q(\Omega,\mathcal{F},P;\Rb^2)$ whose expected value is $[1-\lambda, \quad \lambda]'$.
	Furthermore, $\Ac_{\bfrho_\lambda}$ is equal to the convex subdifferential of $\bfrho_\lambda ([0, \;0]')$. 
\end{proposition}
The proof of Proposition \ref{prop:dr} is provided in Appendix \ref{sec:proof_dual}.
Using a more compact notation, (\ref{eq:dual}) can be written
$$\bfrho_\lambda (X_T)=\sup_{\zeta\in \Ac_{\bfrho_\lambda}} \{-\Eb[\zeta_1 r_T  + \zeta_2 \esg_T]\}.$$

Proposition \ref{prop:lambda0-1} addresses the marginal ESG-coherent risk measure when $\lambda=1$ or $\lambda=0$.

\begin{proposition} \label{prop:lambda0-1}
	If $\bfrho_\lambda$ is an ESG-coherent risk measure, then
	$$\bfrho_0([r_T, \; \esg_T]')=\bfrho_0([r_T, \; 0]'),$$
	$$\bfrho_1([r_T, \; \esg_T]')=\bfrho_1([0, \; \esg_T]').$$
\end{proposition}
\begin{proof}
	For $\lambda=0$, we know from the dual representation that $\zeta_2=0$ a.s., since it has zero expected value and is non-negative. Hence, 
	$$\bfrho_0 ([r_T, \; \esg_T]')  =\sup_{\zeta\in \Ac_0} \left\{-\int_\Omega \zeta_1(\omega) r_T (\omega) P(d\omega)\right\}=\bfrho_0([r_T, \; 0]').$$ 
	Analogously, for $\lambda=1$ we have 
	$$\bfrho_1 ([r_T, \; \esg_T]') =\sup_{\zeta\in \Ac_1} \left\{-\int_\Omega \zeta_2(\omega) \esg_T (\omega) P(d\omega)\right\}=\bfrho_1 ([0, \; \esg_T]').$$
\end{proof}

In other words, Proposition \ref{prop:lambda0-1} states that the risk for an investor with $\lambda=0$ is not affected by the ESG score of the asset, while the risk for an investor with $\lambda=1$ is not affected by the monetary returns.

\subsection{Hedging risk by investing in ESG safe assets}\label{sec:risk_free}

To clarify both the structure and the practical use of ESG-coherent risk measures, we analyze how ESG-oriented investors may hedge a risky position by allocating capital to an ESG safe asset, better clarifying the interpretation of the translation invariance property in the bivariate setting. We will consider the case of an investor or an asset manager that is subject to the budget constraint, and thus to invest in an ESG safe asset needs to disinvest an equal amount from the risky asset, following standard treatments in the risk-measure literature \cite[see, e.g. ][Chapter 6.4.4]{rachev2011risk}. In a traditional univariate framework, a safe asset by definition has a deterministic payoff;\footnote{
	Note that a position can have a risk equal to zero and not be a safe asset. Similarly, an asset with a deterministic payoff (i.e., a safe asset) can have a risk that is different from zero. We can understand this point better by considering a position with a return $r_T$ and risk $\rho(r_T)=m$. If $\rho(\cdot)$ is a coherent risk measure, by axiom TI we have $\rho(r_T+m)=0$. That is, a position with a return $r_T'=r_T+m$ has zero risk (but its returns are not necessarily constant).}
its return is a constant $\rft\in \mathbb{R}$. If we define the risk on a univariate random variable that represents  returns, axioms TI and PO state that the risk of a portfolio composed of the safe asset and a risky position $r_T$ is
\begin{equation}\label{eq:TI_returns}
	\rho((1-w)r_T + w \rft) = (1-w)\rho(r_T) - w \rft,
\end{equation}
where $w\in [0,1]$ is the weight of the safe asset in the portfolio. More formally, we address the problem of an investor who is willing to reduce the risk of a position to an acceptable level $\kappa$ by creating a portfolio consisting of the risky  position, and of the smallest possible amount of the safe asset. The motivation is, for instance, to satisfy requirements imposed by regulators or by the institutional mandate. Formally the problem is 

\begin{align}
w^* = \arg\min_{w} & (w)\notag\\
\text{s.t. } & \rho((1-w)r_T + w \rft)\leq \kappa, \notag\\
& 0 \leq w \leq 1.\label{eq:opt_rf_univariate}
\end{align}

We assume that $-\rft< \kappa\leq \rho (r_T)$. Since the risk of a portfolio is an affine function of $w$ as shown in \eqref{eq:TI_returns}, the solution of \eqref{eq:opt_rf_univariate} is: 
\begin{equation}\label{eq:rfw}
	w^* = \dfrac{\rho(r_T)-\kappa}{\rho(r_T)+\rft}.
\end{equation}
That is, the risky position can be hedged by constructing a portfolio that contains the safe asset having weight $w^*$.\footnote{
	If the risk measure were defined in terms of final net worth rather than returns,
	to hedge a risky position with risk $m$, it would be necessary to add a cash position.
	For a broader discussion of the interpretation of the axioms expressed in terms of returns rather than the final net worth,
	see \citet[][Chapter 6]{rachev2011risk}.}

In the ESG-oriented investing framework, an ESG safe asset is a position having constant values for both monetary and ESG components. For illustrative purposes, we can postulate the existence of several types of such ESG safe assets, distinguished by different combinations of constant ESG and monetary return.

Consider first the case in which only one type of ESG safe asset is available in the market. The variable $\sa_T \in \mathbb{R}^2$ denotes the constant return and constant $\esg$ for the period $T$ of the safe asset:
	\begin{equation}
		\sa_T:=\begin{bmatrix}
			\rft\\ \rftesg
		\end{bmatrix}. \notag
	\end{equation}
	
We know that by axiom LH-M, for an ESG-coherent risk measure, the risk of this ESG safe asset is $\bfrho_{\lambda}(\sa_T) = -((1-\lambda)\rft + \lambda \rftesg)$. The problem of hedging the risk of an asset with bivariate return $X_T$ is analogous to the univariate case: an investor with a given $\lambda$ wants to construct a portfolio with ESG-risk smaller or equal than $\kappa$ by creating a portfolio with the risky asset, and the smallest possible amount of the ESG safe asset (i.e. minimizing its weight in the portfolio):

\begin{align}
w_\lambda^* = \arg\min_{w} & (w)\notag\\
\text{s.t. } & \bfrho_\lambda((1-w)X_T + w \sa_T)\leq \kappa, \notag\\
& 0 \leq w \leq 1.\label{eq:opt_rf_ESG}
\end{align}

We assume that $-(1-\lambda)\rft - \lambda \rftesg< \kappa\leq \bfrho_\lambda(X_T).$ The solution is 
\begin{equation}
	w_\lambda^* = \dfrac{\bfrho_\lambda(X_T)-\kappa}{(1-\lambda)\rft+ \lambda \rftesg+\bfrho_\lambda(X_T)}.
\end{equation}
We underline that $w_\lambda^*$ is unique for each investor and its value varies with the parameter $\lambda$. In practice, such an ESG safe asset could be achieved by the investor making a guaranteed loan to an institution (either a for-profit company, a government, or a non-profit institution) that has a positive and stable environmental or social impact, which generates an interest $\rft$ for the investor. 

We discuss three special cases characterized by specific ESG safe assets. With the exception of case 3, we assume that the ESG safe assets have non-negative return and ESG.

\begin{enumerate}[align=left, leftmargin = *]
\item \textbf{The only ESG safe asset is a \textit{pure monetary safe asset}}, that has the following bivariate return:
	\begin{equation}
		\sa_{\cash,T}:=\begin{bmatrix}
		\rft\\0
		\end{bmatrix}, \notag
	\end{equation}

	where $\rft\geq0$. 
	An example of this is a risk-free, zero-coupon bond issued by a governmental institution not associated with
	a specific ESG profile.\footnote{
	Rating agencies are starting to compute ESG scores for countries as well, although the criteria are different from those used to calculate companies' scores. The identification of pure monetary and pure ESG safe assets will be a significant challenge for practitioners and scholars.}
	An ESG investor can hedge the risk of a position $X_T$ by constructing a portfolio composed of the \textit{pure monetary safe asset} in proportion $w_\lambda^*$ and investing the rest in the risky asset, 
	with  $$w^*_\lambda = \dfrac{\bfrho_\lambda(X_T) - \kappa}{(1-\lambda)\rft  + \bfrho_\lambda(X_T)}.$$
	We note that if $\lambda=1$ (i.e., an investor cares exclusively about the ESG component) $\bfrho_1(\sa_{\cash,T})=0$.

\item \textbf{The only ESG safe asset is a \textit{pure ESG safe asset}} with bivariate return
	\begin{equation}
	\sa_{\esg,T}:=\begin{bmatrix}
	0\\ \rftesg
	\end{bmatrix}. \notag
	\end{equation}
	The analysis of this case is symmetrical to that for the pure monetary safe asset.

\item \textbf{There exists an ESG safe asset having a monetary return of --100\%.}
	We consider the special case of an ESG safe asset represented by
	\begin{equation}
		\sa_{GRANT,T}:=\begin{bmatrix}
			-1\\  \rftesg
			\end{bmatrix}. \notag
	\end{equation}
    The financial component corresponds to a deterministic loss of the entire investment, while the ESG component is positive and constant. As a conceptual benchmark, this position may be interpreted as a grant to a charitable organization generating a fixed ESG contribution. Such an asset is clearly not a relevant investment opportunity for a risk-return investor with $\lambda=0$.\footnote{Assuming the absence of tax benefits associated to the donation, or other forms of monetary gains.} On the other hand, for an ESG-oriented investor with $\lambda>0$ it could be rational to invest in such asset.
	The measured risk of such an asset is $\bfrho_{\lambda}(\sa_{GRANT,T}) = (1-\lambda) - \lambda \rftesg$;
	that is, for an investor with a given $\lambda$ the risk is negative if $\lambda >\tilde{\lambda} =  1/(1+\rftesg)$.
	For a ``$\lambda<\tilde{\lambda}$'' investor, such an asset provides no opportunity to hedge a risky position; however for a ``$\lambda>\tilde{\lambda}$'' investor, such an ESG safe asset provides a meaningful
	hedging tool through the donation of a wealth fraction:
	$$
	 w^*_\lambda = \dfrac{\bfrho_\lambda(X_T) - \kappa}{\lambda-1+ \lambda \rftesg+\bfrho_\lambda(X_T)}
	$$
	 to a project with a positive environmental or social impact.
\end{enumerate}

We can also study the case with multiple ESG safe assets available in the market. In such case, while it appears that an investor could choose multiple ESG safe assets, a simple analysis shows that each investor will choose only a single ESG safe asset among the available ones.
If  $n$ ESG safe assets with returns $\sa_{i,T}, \;i=1,\dots,n$, are available, an investor can hedge a risky position with bivariate return $X_T$ by
creating a portfolio with weight $w_{0}$ for the risky asset and $w_i, \;i=1,\dots,n$ for the safe assets, with $\sum_{i=0}^{n}w_i=1$ such that the risk is equal to $\kappa$. We can formulate an optimization problem analogous to \ref{eq:opt_rf_ESG}, with the difference that the objective function to maximize is the weight of the risky asset $w_{0}$ (that is equivalent to minimizing the sum of the portfolio weights invested in ESG safe assets). Formally:
\begin{align}
	\max_{w_0,w_1,w_2,\dots,w_{n}} & (w_{0})\notag\\
	\text{s.t. } & \bfrho_{\lambda}(w_{0}X_T + w_1 \sa_{1,T}+ w_2 \sa_{2,T}+ \dots + w_n \sa_{n,T}){\leq \kappa,}\notag\\
	& \sum_{i=1}^{n+1}w_i=1, \quad 0 \leq w_i \leq 1, \quad i = 1,\dots, n+1.
\end{align}

The constraint set defines a convex feasible subset of $\Rb^n$ for the investment allocations, $w_i$. The solution is not a diversified portfolio;  only one ESG safe asset is selected.
This follows from the fact that, due to LH-M and TI-M, the risk of a sum of ESG safe assets is the sum of their risk: $\bfrho_{\lambda}(\delta\sa_{i,T} + (1-\delta)\sa_{j,T})=\delta\bfrho_{\lambda}(\sa_{i,T}) + (1-\delta)\bfrho_{\lambda}(\sa_{j,T}),\; \delta\in[0,1]$. For any pair of ESG safe assets $i$ and $j$, the risk of their convex combination is always greater or equal than $\min(\bfrho_{\lambda}(\sa_i),\bfrho_{\lambda}(\sa_j)$. Hence, risk minimization is achieved by selecting the ESG safe asset with minimal risk, and not by forming portfolios of safe assets. This uniqueness property, that is a direct consequence of ESG-coherence of $\bfrho_\lambda$, parallels the univariate setting, where, in the presence of multiple safe assets, the optimal choice is the cheapest one.\footnote{An exception is when two or more ESG safe assets with exactly the same risk are available.
	In such a case, they are indistinguishable to the investor in terms of risk, and either can be chosen.}
Once the ESG safe asset with the lowest risk is identified, the problem then is the same as (\ref{eq:opt_rf_ESG}) where only one ESG safe asset was available. Note however that the choice of ESG safe asset is different for each investor, as the optimization is dependent on the investor's value for $\lambda$.

\begin{remark}
	In general, for risk minimization under a budget constraint, the choice of the ESG safe asset is not influenced by the characteristics of the risk of the position, it is only influenced by the availability and price of ESG safe assets and the $\lambda$ preference of the investor.
\end{remark}

\subsection{Examples of ESG-coherent risk measures} \label{sec:examples}

After discussing ESG-coherent risk measures in general, we present two example approaches for extending univariate risk measures to a bivariate setting. In particular, starting from a univariate coherent risk measure $\rho(r_T)$, we identify the ESG-coherent risk measures $\esgd\bfrho_{\lambda}(X_T)$, and $\esgd\bfrho_{\lambda}^l(X_T)$. The measures encompass $\lambda \in [0,1]$ as a parameter, thus they are families of bivariate risk measures suitable for investors having differing ESG ``inclinations''.
In Section \ref{sec:avar} we apply these approaches to the well-known coherent risk measure, average value at risk (AVaR), resulting in two versions of coherent ESG-AVaR risk measures. These examples extend naturally to the ESG formulation without changing their convex-analytic structure.
In Section \ref{sec:varvol} we apply these approaches to two non-coherent risk measures, variance and volatility, producing ESG extensions that are not ESG-coherent.

Our first approach to generalizing a univariate risk measure $\rho(r_T)$ utilizes a linear combination of $r_T$ and $\esg_T$. For $X_T = [ r_T, \; \esg_T ]'$
\begin{align}
\esgd\bfrho_{\lambda}\left( X_T \right) & := \rho\big((1-\lambda)r_T + \lambda\esg_T \big). \label{eq:esgrho}
\end{align}

\begin{proposition} \label{prop:esgrho}
If $\rho(\cdot)$ is a univariate coherent risk measure, then $\esgd\bfrho_{\lambda}(\cdot)$ is an ESG-coherent risk measure.
\end{proposition}
\begin{proof}
Since the right-hand-side of \eqref{eq:esgrho} involves a direct application of $\rho(\cdot)$,
the extended function $\esgd\bfrho_{\lambda}(\cdot)$ inherits the properties SUB-M and PH-M.
It is straightforward to show that axiom MO of the univariate function implies MO-M.
Consider two vectors $X_{1,T}$, $X_{2,T} \in \mathcal{X}_2$.
If $(r_{1,T} \leq r_{2,T} \wedge \esg_{1,T}\leq \esg_{2,T})$ a.s., then from the monotonicity of $\rho(\cdot)$,
we have

$$
	\esgd\bfrho_\lambda(X_{1,T}) =    \rho((1-\lambda)  r_{1,T} + \lambda \esg_{1,T})
						\geq \rho((1-\lambda)  r_{2,T} + \lambda \esg_{2,T})
						= \esgd\bfrho_\lambda(X_{2,T}).
$$

Finally, using the translation invariance of $\rho(\cdot)$, we can prove LH-M. 
\end{proof}

We note the following properties of $\esgd\bfrho_\lambda(\cdot)$.
\begin{itemize}[align=left, leftmargin = *]
	\item If $\rho(r_T)$ is convex, then $\esgd\bfrho_{\lambda}([r_T,\esg_T]')$ is a convex function of $\lambda$.
	\item For $\lambda = 0$ and $\lambda = 1$, $\esgd\bfrho_\lambda([r_T,\esg_T]')$ is equal to the univariate $\rho(r_T)$ computed on returns alone or $\rho(\esg_T)$ computed on ESG scores alone, respectively.
\end{itemize}

Our second approach utilizes a linear combination of univariate risk measures,
\begin{equation}
	\esgd\bfrho^l _\lambda\left( \left[ \begin{matrix}r_T\\ \esg_T\end{matrix}\right] \right) := (1-\lambda)\rho(r_T)+\lambda\rho(\esg_T). \label{eq:esgrhol}
\end{equation}

It is straightforward to show that $\esgd\bfrho^l _\lambda(\cdot)$ is an ESG-coherent risk measure if $\rho(\cdot)$ is coherent because axioms SUB-M, PH-M, and MO-M follow from the respective univariate axioms, while LH-M follows from TI.

We  note the following properties of $\esgd\bfrho^l_\lambda(\cdot)$.
\begin{itemize}[align = left, leftmargin = *]
	\item The measure $\esgd\bfrho^l _\lambda([r_T,\esg_T]')$ is more conservative than $\esgd\bfrho _\lambda([r_T,\esg_T]')$ as it is linear in $\lambda$ and, hence, always greater than or equal to $\esgd\bfrho _\lambda([r_T,\esg_T]')$.
	\item  $\esgd\bfrho^l _\lambda([r_T,\esg_T]')$ is equivalent to $\esgd\bfrho _\lambda([r_T,\esg_T]')$ for the case of perfect co-monotonicity between $r_T$ and $\esg_T$ (i.e. when no diversification between the two is possible). In this sense, we can consider it a worst-case measure of $\esgd\bfrho _\lambda([r_T,\esg_T]')$.
	\item For the limiting cases $\lambda=0$ and $\lambda=1$, we have $\esgd\bfrho^l _\lambda([r_T,\esg_T]') = \esgd\bfrho _\lambda([r_T,\esg_T]')$.
	\item For $\lambda=0$ or $\lambda=1$,  $\esgd\bfrho^l _\lambda([r_T,\esg_T]')$ corresponds to the univariate risk measure $\rho(X_T)$ computed on just the monetary part or just the ESG part, respectively.
\end{itemize}

\subsubsection{ESG-AVaR}\label{sec:avar}

We demonstrate the two approaches presented in equations \eqref{eq:esgrho} and \eqref{eq:esgrhol} to develop ESG-coherent risk measures based on $\avar$.
Given a random bivariate vector $X_T=[r_T, \;\esg_T]' \in \Xc_2=\mathcal{L}_1(\Omega,\mathcal{F},P;\Rb^2)$,
the first measure, given by \eqref{eq:esgrho}, is\footnote{
	See \cite{ogryczak2002dual,rockafellar2002conditional} for the extremal representation.}
\begin{align}
\esgavar\left( \left[ \begin{matrix}r_T\\ \esg_T\end{matrix}\right] \right) & := \avar_\tau\left((1-\lambda)r_T+\lambda\esg_T \right) \notag \\
	& = \inf_{\beta \in \Rb} \left\{ \frac{1}{1-\tau} \Eb\left[\left(\beta - ((1-\lambda)r_T+\lambda\esg_T) \right)^+\right] - \beta\right \}, \label{eq:esgavar}
\end{align}
where $(a)^+$ denotes max$(a,0)$.
Following the discussion on \eqref{eq:esgrho}, we conclude that $\esgavar$ is coherent. It is similar to the multivariate expected shortfall presented by \cite{ekeland2012comonotonic} (although their measure lacks the parametrization using $\lambda$ allowing investor-specific ESG preferences). 
For $\lambda=0$ the measure reduces to the univariate AVaR. Since $\esgavar(\cdot)$ is computed on univariate data (i.e., as a linear combination of $r$ and $\esg$), numerical applications using $\esgavar(\cdot)$ do not present any particular challenge; it is possible to fully utilize existing procedures developed for $\avar$ for risk estimation, portfolio optimization, and risk management. (See, e.g. \citealp{shapiro2021lectures}).

The dual representation of $\esgavar(X_T)$ is
\begin{equation}
	\esgavar\left( \left[ \begin{matrix}r_T\\ \esg_T\end{matrix}\right] \right)
		= \sup_{[\zeta_1, \;\zeta_2]' \in \Ac_{\esgavar(X_T)}} \left( - \Eb[r_T\zeta_1 + \esg_T \zeta_2] \right),
\end{equation}
where 
\begin{equation}
\Ac_{\esgavar} = \left\{[\zeta_1, \;\zeta_2]' \in \mathcal{L}_\infty(\Omega,\mathcal{F},P;\Rb^2):\; [\zeta_1, \; \zeta_2]' =\xi[1-\lambda  ,\; \lambda]'; \; 0\leq \xi\leq \frac{1}{1 -\tau}\; \text{a.s. }; \Eb[\xi]= 1\right\}.
\end{equation}
The derivation of this dual representation is given in Appendix \ref{sec:proof_avar_dual}.

The second measure, given by \eqref{eq:esgrhol}, is
\begin{align}
\esgavar^l\left( \left[ \begin{matrix}r_T\\ \esg_T\end{matrix}\right] \right) & := (1-\lambda)\avar_\tau(r_T)+\lambda\avar_\tau(\esg_T). \label{eq:esgavarl}
\end{align}
$\esgavar^l(\cdot)$ is also ESG-coherent. It can be viewed as the limit of $\esgavar(\cdot)$ in the case of  an asset for which it is not possible to diversify between the monetary and ESG components as they are comonotone. From an economic perspective, $\esgavar^l(\cdot)$ is significant for investors who consider the worst-case scenario in terms of the dependency structure.

The dual representation of $\esgavar^l(X_T)$ is
\begin{equation}
	\esgavar^l\left( \left[ \begin{matrix}r_T\\ \esg_T\end{matrix}\right] \right) = \sup_{[\zeta_1, \; \zeta_2]' \in \Ac_{\esgavar^l(X_T)}} - \Eb[r_T\zeta_1 + \esg_T \zeta_2],
\end{equation}
where
\begin{equation}
\Ac_{\esgavar^l} = \left\{[\zeta_1, \; \zeta_2]'\in \mathcal{L}_\infty(\Omega,\mathcal{F},P;\Rb^2): \Eb[\zeta_1]=1-\lambda; \Eb[\zeta_2]=\lambda;\zeta_1,\zeta_2\geq 0;\zeta_1\leq\frac{1-\lambda}{1-\tau}; \zeta_2\leq\frac{\lambda}{1-\tau}\right\}.
\end{equation}
The derivation of the dual representation of $\esgavar^l(\cdot)$ is also given in Appendix \ref{sec:proof_avar_dual}.

\subsubsection{Non-ESG-coherent measure examples} \label{sec:varvol}

In order to better delineate the boundaries of the axiomatic framework, we discuss here examples of measures that are not ESG-coherent. It is well known that the standard deviation $\sigma$ (the volatility) and the variance $\mathbb{V}$ are not coherent risk measures.
We consider the application of  \eqref{eq:esgrho} and \eqref{eq:esgrhol} to $\sigma$ and $\mathbb{V}$ and show that,
in all cases, the result is an ESG measure that does not satisfy the ESG-coherency axioms.

Given a vector $X_T=[r_T, \; \esg_T]' \in \Xc_2=\mathcal{L}_2(\Omega,\mathcal{F},P;\Rb^2)$, from \eqref{eq:esgrho}
 the ESG variance and ESG volatility are
\begin{align}
	\esgd\mathbb{V}_\lambda(X_T) & := \mathbb{V}[(1-\lambda)r_T + \lambda \esg_T], \\
	\esgd\sigma_\lambda(X_T) & := \sqrt{\esgd\mathbb{V}_\lambda(X_T)}.\label{eq:esgsd}
\end{align}
$\esgd\mathbb{V}_\lambda(\cdot)$ is not ESG-coherent, as it does not satisfy PH-M, MO-M, SUB-M, and LH-M.
$\esgd\sigma_\lambda(\cdot)$ is not ESG-coherent, as it does not satisfy MO-M and LH-M.

Using \eqref{eq:esgrhol}, the corresponding risk measures are
\begin{align}
\esgd\mathbb{V}^l_\lambda(X_T) &:= (1-\lambda)\mathbb{V}[r_T] + \lambda \mathbb{V}[\esg_T], \\
\esgd\sigma^l_\lambda(X_T) & := \sqrt{\esgd\mathbb{V}^l_\lambda(X_T)}.
\end{align}
The former does not satisfy PH-M, MO-M, SUB-M, and LH-M, and the latter does not satisfy MO-M and LH-M.

A summary of the properties of the examples considered in sections \ref{sec:avar} and \ref{sec:varvol} is given in Table \ref{tab:ESG-risk}.
\begin{table}[!h]
\begin{center}
	\caption{Summary of ESG risk measures and their properties. A checkmark indicates that the corresponding property is satisfied.}
	\label{tab:ESG-risk}
	\begin{tabular}{lcccc}\hline 
		\textbf{Risk measure}			& \textbf{SUB-M} &	\textbf{PH-M} &\textbf{MO-M} &\textbf{LH-M} \\\hline
		$\esgavar$					& $\cm$ & $\cm$ & $\cm$ & $\cm$ \\
		$\esgavar^l$					& $\cm$ & $\cm$ & $\cm$ & $\cm$ \\
		$\esgd\mathbb{V}_\lambda$		&	      &		   &		&	     \\
		$\esgd\sigma_\lambda$			& $\cm$ & $\cm$ &		&	     \\
		$\esgd\mathbb{V}^l_\lambda$	&	      &		  &		&	     \\
		$\esgd\sigma^l_\lambda$		& $\cm$ & $\cm$ &		&	     \\\hline
	\end{tabular} 
    \end{center}
\end{table}

\section{ESG reward--risk ratios}\label{sec:RRR}

It is natural to extend the ESG framework to reward--risk ratios (RRRs), used to measure risk-adjusted performance of an investment. Following \cite{cheridito2013reward}, a reward--risk ratio $\alpha(r)$ in a univariate setting is
\begin{equation}
	\alpha(r_T) := \dfrac{\theta(r_T)^+}{\rho(r_T)^+},
\end{equation}\label{eq:RRR}
where $\theta(r_T): \mathcal{X} \rightarrow \mathbb{R} \cup \{\pm \infty\}$  and $\rho(r_T): \mathcal{X} \rightarrow \mathbb{R} \cup \{+ \infty\}$ are reward and risk measures, respectively. 
\cite{cheridito2013reward} identified four conditions desirable for RRRs:

\begin{itemize}[align = left]
	\item[(MO-R)] monotonicity: if $r_{1,T},r_{2,T} \in \mathcal{X}$ and $r_{1,T}\leq r_{2,T}$ a.s., then $\alpha(r_{1,T}) \geq \alpha(r_{2,T})$;
	\item[(QC-R)] quasi-concavity: if $r_{1,T},r_{2,T} \in \mathcal{X}$ and $\delta \in [0,1]$, then $\alpha(\delta r_{1,T} + (1-\delta) r_{2,T}) \geq \min(\alpha(r_{1,T}),\alpha(r_{2,T}))$;
	\item[(SI-R)] scale invariance:  if $r_T \in \mathcal{X}$ and $\delta > 0$ s.t. $\delta r_T \in \mathcal{X}$, then $\alpha(\delta r_T) =  \alpha(r_T)$;
	\item [(DB-R)] distribution-based: $\alpha(r_T)$ only depends on the distribution of $r_T$ under $P$.
\end{itemize}

Following the approach used for risk and reward measures, we introduce ESG reward--risk ratios (ESG-RRRs). Let $X_T=[r_T, \; \esg_T]'$. We define an ESG-RRR $\bfalpha_\lambda: \mathcal{X}_2 \rightarrow \mathbb{R} \cup \{\pm \infty\}$ by
\begin{equation}\label{eq:ESG-RRR}
\bfalpha_\lambda(X_T) := \dfrac{\bftheta_\lambda(X_T)^+}{\bfrho_\lambda(X_T)^+},
\end{equation}
where $\bftheta_\lambda(X_T)$ and $\bfrho_\lambda(X_T)$ are ESG reward and risk measures as defined in Section \ref{sec:risk_measures} and Appendix \ref{sec:reward} using the same preference parameter $\lambda$.\footnote{
	In principle, it is possible for the risk and reward measures to have different values of $\lambda$. We consider the case of values of lambda for both the numerator and the denominator for conciseness.}
The extension of the Cheridito-Kromer conditions for ESG-RRRs are:

\begin{itemize}[align = left]
	\item[(MO-RM)] monotonicity: if $X_{1,T},X_{2,T} \in \mathcal{X}_2$ and ($r_{1,T}\leq r_{2,T} \wedge \esg_{1,T}\leq \esg_{2,T}$) a.s., then $\bfalpha_\lambda(X_{1,T}) \leq \bfalpha_\lambda(X_{2,T})$;
	\item[(QC-RM)] quasi-concavity: if $X_{1,T},X_{2,T} \in \mathcal{X}_2$ and $\delta \in [0,1]$, then $\bfalpha_\lambda(\delta X_{1,T} + (1-\delta) X_{2,T}) \geq \min(\bfalpha_\lambda(X_{1,T}),\bfalpha_\lambda(X_{2,T}))$;
	\item[(SI-RM)] scale invariance:  if $X_T \in \mathcal{X}_2$ and $\delta > 0$ s.t. $\delta X_T \in \mathcal{X}_2$, then $\bfalpha_\lambda(\delta X_T) =  \bfalpha_\lambda(X_T)$;
	\item [(DB-RM)] distribution-based: $\bfalpha_\lambda(X_T)$ only depends on the distribution of $X_T$ under $P$.
\end{itemize}

Verification of DB-RM depends on the bivariate distribution of $X_T \in \mathcal{X}_2$ and not on the univariate distribution of the returns.

\begin{proposition} \label{prop:RRR}
	The ESG-RRR  \eqref{eq:ESG-RRR}, where $\bftheta_\lambda(X_T)$ is an ESG-coherent risk measure and $\bfrho_\lambda(X_T)$ is an ESG-coherent reward measure (Appendix \ref{sec:reward}) satisfies conditions MO-RM, QC-RM and SI-RM.
\end{proposition}
\begin{proof}
	MO-RM follows from the monotonicity of $\bfrho_\lambda(X_T)$ and $\bftheta_\lambda(X_T)$.
	QC-RM follows from the convexity of $\bfrho_\lambda(X_T)$ and the concavity of $\bftheta_\lambda(X_T)$.
	SI-RM follows from the corresponding properties of $\bfrho_\lambda(X_T)$ and $\bftheta_\lambda(X_T)$.
\end{proof}

\begin{remark}
 Proposition \ref{prop:RRR} does not guarantee that the DB-RM property is satisfied since, as discussed in \citep{cheridito2013reward}, risk and reward measures may not depend on the distribution of returns and ESG under a single probability measure, and in such case the distribution-based property is not satisfied. This is for example the case of distributionally robust reward--risk ratios, which consider the fact that agents do not know with certainty the distribution of the random variables, and consider a whole family of distributions \citep[see e.g.][]{liu2017distributionally}.
\end{remark}

\subsection{Examples of ESG reward--risk ratios}

We present six examples of ESG reward--risk ratios derived from RRRs commonly used in the literature. The ratios are obtained by generalizing the univariate reward--risk ratios using the approach described by \eqref{eq:esgrho}. 
Note that, as in the case of risk measures, the definition of a reward--risk ratio can be based on several alternative specifications of the random variable $X_T=[r_T,\;\esg_T]'$. In particular, the ratios can be computed using the returns, the excess returns over a risk-free rate, the final wealth, profit and/or losses, etc. The same logic applies to the ESG component. Here, we only provide illustrative guidance concerning which approach to use in practice; the choice depends on the specific needs of the practitioner or regulator who uses these measures.

\begin{itemize}[align = left, leftmargin = *]
	\item \textbf{ESG Sharpe ratio}. The Sharpe ratio is the ratio between the excess return of an asset and its standard deviation  over a period of time.  The ESG Sharpe ratio is
	\begin{equation}
	\esgd\sr_\lambda(X_T) := \dfrac{\esgd\mu_\lambda(X_T-\sa_T)}{\esgd\sigma_\lambda(X_T)},
	\end{equation}
	where
	\begin{equation}
	\esgd\mu_\lambda(X_T) := (1-\lambda)\mathbb{E}[r_T] + \lambda \mathbb{E}[\esg_T],
	\end{equation}	
 the ESG standard deviation $\esgd\sigma_\lambda(X_T)$ is given by (\ref{eq:esgsd}), and $\sa_T = [\rft, \;\rftesg]'$ is the bivariate return of an ESG safe asset.
 $\esgd\mu_\lambda$ is an ESG-coherent reward measure, while $\esgd\sigma_\lambda$ is not an ESG-coherent risk measure.  
	The ESG Sharpe ratio satisfies QC-RM, SI-RM, and DB-RM, but it does not satisfy condition MO-RM.\footnote{
	The SI-RM property applies if $\sa=[0, \; 0]'$ or if $X$ is intended to represent a vector of excess returns over the risk-free rate.}
	\item \textbf{ESG Rachev ratio}.
	The univariate Rachev ratio is defined as \citep{biglova2004different}
	\begin{equation}
	\rr_{\beta,\gamma}(r_T) := \dfrac{\avar_\beta (-r_T)}{\avar_\gamma(r_T)}.
	\end{equation}
	We generalize to an ESG-RR by replacing $\avar$ with $\esgd \avar$. Let $X_T \in \mathcal{X}_2$. Then
	
	\begin{equation}
	\esgd\rr_{\beta,\gamma,\lambda}(X_T) := \dfrac{\esgd\avar_{\lambda,\beta}(-X_T)}{\esgd\avar_{\lambda,\gamma}(X_T)}.
	\end{equation}
	
	Note that $\esgd\rr$ satisfies MO-RM, SI-RM and DB-RM but, as in the univariate case, fails to satisfy QC-RM since the numerator is not concave.
	\item \textbf{ESG STAR ratio}. When $\beta=0$, the ESG Rachev ratio becomes an ESG generalization of the stable tail-adjusted return ratio,
	\begin{equation}
	\esgd\starr_{\alpha,\lambda}(X_T) := \dfrac{\esgd \mu_{\lambda}(X_T)}{\esgd\avar_{\lambda,\alpha}(X_T)}.
	\end{equation}
	As a special case of the ESG Rachev ratio, $\esgd\starr_{\alpha,\lambda}$ satisfies conditions MO-RM, SI-RM and DB-RM;
	as the numerator is linear, it also satisfies QC-RM.
	\item \textbf{ESG Sortino--Satchell ratio}.
	The univariate Sortino--Satchell ratio \citep{sortino2001managing} is defined as $\mathbb{E}[r_T]^+/||r_T^-||_p$,
	where $||r_T||_p=\left(\int_{-\infty}^{\infty} |x|^p f_r(x)dx\right)^{1/p}$  
	and $r_T \sim f_r(x)$.
	We extend this measure to the bivariate ESG setting by\footnote{
		The definition provided here assumes a target return of 0, similarly to \cite{cheridito2013reward}.}
	\begin{equation}
	\esgd \ssr_\lambda(X_T) = \dfrac{\left(\esgd \mu_\lambda(X_T)\right)^+}{||Y_T^-||_p},
	\end{equation}
	where $Y_T= (1-\lambda)r_T + \lambda \esg_T$ and $Y_T \sim f(y)$.
	This measure satisfies all four ESG-RRR conditions. 
    
	The proposed formulation assumes a required rate of return and ESG target of zero.
	We can introduce a non-zero target by subtracting a bivariate vector (e.g. the bivariate return of an ESG safe asset $\sa_T = [\rft, \;\rftesg]'$) from the numerator before applying the positive operator and using $\tilde{Y}_T = (1-\lambda)(r_T-\rft) + \lambda (\esg_T-\rftesg)$ in the denominator in place of $Y_T$. In such a case, this measure no longer satisfies SI-RM.
	\item \textbf{ESG Omega ratio}. Defining $Y_T= (1-\lambda)r_T + \lambda \esg_T$, and $F(y)$ as the cumulative distribution function of $Y_T$,
	the ESG version of the Omega ratio \citep{keating2002universal} is
	\begin{equation}
	\esgd \OR_\lambda(X_T) = \dfrac{\int_\tau^{\infty}[1-F(y)]dy^+}{\int_{-\infty}^\tau F(y)dy}.
	\end{equation}
	Equivalently \citep[see][]{farinelli2008sharpe},
	\begin{equation}
	\esgd \OR_\lambda(X_T) = \dfrac{\Eb[(Y_T-\tau)^+]}{\Eb[(Y_T-\tau)^-]}.
	\end{equation}
	This measure satisfies MO-RM and DB-RM. It does not satisfy QC-RM, and SI-RM is only satisfied if $\tau=0$.
	\item \textbf{ESG Farinelli--Tibiletti ratio}. This ratio aims to take into account the asymmetry in the return distribution; it generalizes the Omega ratio. By defining $Y_T= (1-\lambda)r_T + \lambda \esg_T$, the ESG version of the ratio can be defined as
	\begin{equation}
\esgd \ftr_{\lambda,m,n,p,q}(X_T) = \dfrac{||(Y_T-m)^+||_p}{||(n-Y_T)^+||_q},
\end{equation}
	where $m,n \in \Rb$ and $p,q>0$. This ratio satisfies MO-RM and DB-RM. It also satisfies SI-RM if $n=m=0$. Mirroring the univariate case, it does not satisfy QC-RM (see the example in \citealp[][Section 4.3]{cheridito2013reward}).
\end{itemize}

Table \ref{tab:ESG-RRR} summarizes the properties satisfied by each of these six ratios.
{\begin{table}[!h]
\begin{center}
	\caption{Summary of ESG reward--risk ratios and their properties. A checkmark indicates that the corresponding property is satisfied.}
	\label{tab:ESG-RRR}
	\begin{tabular}{lcccc}
		\hline 
		\textbf{Ratio}		 & \textbf{MO-RM} & \textbf{QC-RM} & \textbf{SI-RM} &\textbf{DB-RM} \\ \hline 
		ESG Sharpe		 &			     & $\cm$		   & $\cm$		& $\cm$\\
		ESG Rachev		 & $\cm$		     &				  & $\cm$		& $\cm$\\
		ESG STAR			 & $\cm$		     & $\cm$		  & $\cm$		& $\cm$\\
		ESG Sortino--Satchell & $\cm$		     & $\cm$		  & $\cm$		& $\cm$\\
		ESG Omega		 & $\cm$		     &				 & $\cm$		& $\cm$\\
		ESG Farinelli--Tibiletti  & $\cm$	     &				 & $\cm$		& $\cm$\\\hline
	\end{tabular} 
    \end{center}
\end{table}}

\section{Empirical analysis}\label{sec:empirical}

We present an empirical analysis in which we estimate a set of daily ESG risk measures and ESG reward--risk ratios. We use such measures to rank equity assets from the Dow-Jones Industrial Average (DJIA) index, and the main goal is to illustrate how rankings vary with the preference parameter $\lambda$. We use the log-returns computed on adjusted close prices downloaded from Factset, and for the ESG component we consider the Factset Truvalue scores. These ESG scores, in contrast to scores by other data providers, are not computed on the basis of balance sheet items and metrics reported by companies, but instead they are obtained by analysing news and documents, performing sentiment analyses on news and event tracking on sustainability-related events with artificial intelligence techniques. Such scores are based on a large number of sources, and reflect the public perception of the investors on the sustainability of companies. One unique feature of these scores is that they are updated with daily frequency, allowing the dynamics of the sustainability of each company to be tracked. In contrast, other data providers update ESG scores with quarterly or yearly frequency. The Truvalue scores represent therefore a desirable dataset for showcasing ESG-coherent risk measures and reward ratios. Truvalue scores are provided in two variants: \textit{Pulse Scores} and \textit{Insight Scores}, where the former focus more on near-term performance changes, and the latter considers the longer-term track record. In this analysis we consider Pulse Scores, as they provide a more dynamic measure of the sustainability profile of the companies. The high-frequency of the Truvalue dataset however comes at the cost of a less transparent scoring methodology compared to competitors due to the reliance to AI techniques.

 \begin{figure}[!h]
	\centering
	\includegraphics[width=1\linewidth]{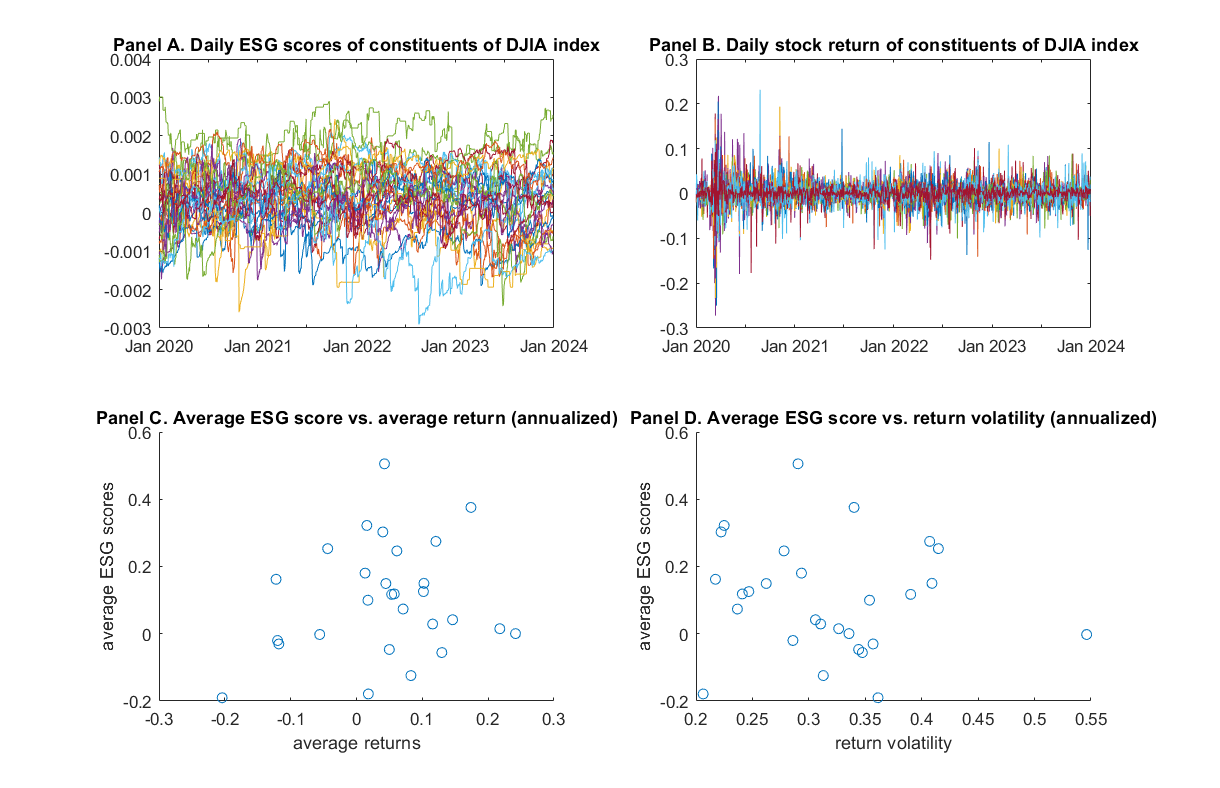}
	\caption{Panel A: Evolution of the daily ESG score of the constituents of the index. Panel B: Evolution of stock return. Panel C: Scatterplot of average annualized daily returns and the average annualized ESG score for each asset. Panel D: Scatterplot of the annualized daily standard deviation of returns and the annualized average ESG score.}
	\label{fig:plot_ESG_prices}
\end{figure}

We consider 28 constituents of the DJIA index as of December 31, 2023. We use the daily time series of log-returns ($r_T$) and Truvalue Pulse Scores ($\esg_T$) in the period from January 1, 2020 to December 31, 2023. The daily Pulse Score is normalized between $-1/c$ and $1/c$, where $c = 252$ (the number of trading days in a year). The normalization leads to a broadly similar scale of average returns and average ESG (see Table \ref{tab:Description}) The scaling of the ESG score, as long as it is consistent across assets, does not affect the generality of the results, as it can be compensated by changing the scaling of the subjective parameter $\lambda$. One concern with the dataset is that, when ESG scores are strictly bounded, some axiomatic conditions may lose their standard economic interpretation. In applications where this issue is relevant, one may adopt the modeling convention of working with a monotone transformation of the scores that has unbounded support, thereby restoring a more formal setting without altering their ordinal content.

Figure \ref{fig:plot_ESG_prices}, illustrates the series of ESG scores for the constituents of the index (Panel A), and the log-return of the assets over time (Panel B). Panels C and D report the average ESG score across the entire period in relation to the average returns and the return volatility, respectively; looking at this aggregate data, the relation between the ESG score and the financial performance seems limited. Higher ESG scores seem to correlate slightly with higher returns, but the relation does not appear to be strong (correlation: $0.18$). We see a small negative correlation between ESG scores and volatility (correlation: $-0.12$).
 
 \begin{table}[!h]	
 \begin{center}
 	\begin{scriptsize}
 		\caption{Company name, ticker symbol and GICS sector; annualized mean and standard deviation of daily 
 			returns ($\E[r]$, $\sigma(r)$) and ESG scores ($\E[\esg]$, $\sigma(\esg)$), correlation of daily returns
 			and ESG score ($\rho(r,\esg)$) for the period 2020-2023.}
 			\label{tab:Description}
 		{\scriptsize
 			\begin{tabular}{l c rrrrr}
 				\toprule
 				Name & GICS Sector & $\E[r]$ & $\sigma(r)$ & $\E[\esg]$ & $\sigma(\esg)$ & $\rho(r,\esg)$\\
 				\midrule
 				Apple Inc.                                   &  IT 			&  $ 0.242$  &  $0.336$  & $         0.000$ & $0.005$ & $-0.025$ \\   
 				Amgen Inc.                                 &  Health Care	&  $ 0.045$  &  $0.262$  & $ 0.149$ & $0.012$ & $-0.002$ \\   
 				American Express Co.            &  Financials		&  $ 0.102$  &  $0.409$  & $  0.150$ & $ 0.010$ & $-0.015$ \\   
 				Boeing Company                      &  Industrials		&  $-0.056$  &  $0.546$  & $-0.002$ & $0.009$ & $ 0.041$ \\   
 				Caterpillar Inc.                          &  Industrials		&  $ 0.174$  &  $ 0.340$   & $ 0.376$ & $0.006$ & $-0.019$ \\   
 				Salesforce, Inc.                         &  IT			&  $ 0.121$  &  $0.407$  & $ 0.275$ & $0.007$ & $-0.035$ \\   
 				Cisco Systems, Inc.                  &  IT 			&  $ 0.013$  &  $0.293$  & $  0.180$ & $0.008$ & $  0.060$ \\   
 				Chevron Corporation             &  Energy		&  $ 0.053$  &  $ 0.390$  & $ 0.117$ & $0.007$ & $ 0.018$ \\   
 				Walt Disney Company           &  Comm. Services	&  $-0.118$  &  $0.357$  & $ -0.030$ & $0.008$ & $ 0.035$ \\   
 				Goldman Sachs Gr., Inc.         &  Financials		&  $  0.130$  &  $0.347$  & $-0.056$ & $ 0.010$ & $ -0.030$ \\   
 				Home Depot, Inc.                    &  Consumer Discr.	&  $ 0.116$  &  $ 0.310$  & $ 0.029$ & $ 0.010$ & $-0.008$ \\   
 				Honeywell Int. Inc.                &  Industrials		&  $ 0.042$  &  $ 0.290$  & $ 0.505$ & $0.007$ & $ 0.047$ \\   
 				IBM Corporation 		 &  IT				&  $ 0.061$  &  $0.278$  & $ 0.246$ & $0.005$ & $ 0.011$ \\   
 				Intel Corporation                    &  IT			&  $-0.044$ &  $0.415$  & $ 0.253$ & $0.008$ & $ 0.007$ \\   
 				Johnson \& Johnson              &  Health Care		&  $ 0.018$  &  $0.206$  & $-0.179$ & $0.008$ & $ 0.024$ \\   
 				JPMorgan Chase \& Co.       &  Financials		&  $  0.050$  &  $0.344$  & $-0.047$ & $0.009$ & $-0.022$ \\   
 				Coca-Cola Company             &  Consumer Staples&  $ 0.016$  &  $0.225$  & $ 0.322$ & $0.004$ & $ 0.037$ \\   
 				McDonald's Corp.                  &  Consumer Discr. 	&  $ 0.102$  &  $0.247$  & $ 0.125$ & $0.006$ & $ -0.010$ \\   
 				3M Company                           &  Industrials		&  $ -0.120$  &  $0.286$  & $ -0.020$ & $0.013$ & $  0.010$ \\   
 				Merck \& Co., Inc.                 &  Health Care 		&  $ 0.057$  &  $0.241$  & $ 0.118$ & $0.005$ & $ 0.023$ \\   
 				Microsoft Corporation         &  IT				&  $ 0.218$  &  $0.326$  & $ 0.015$ & $0.006$ & $ 0.011$ \\   
 				NIKE, Inc. Class B                   &  Consumer Discr. 	&  $ 0.017$  &  $0.354$  & $  0.100$ & $0.008$ & $-0.007$ \\   
 				Procter \& Gamble Co.       &  Consumer Staples	&  $  0.040$  &  $0.222$  & $ 0.303$ & $0.009$ & $-0.024$ \\   
 				Travelers Companies, Inc.&  Financials		&  $ 0.083$  &  $0.313$  & $-0.124$ & $0.014$ & $-0.021$ \\   
 				UnitedHealth Group Inc.  &  Health Care		&  $ 0.146$  &  $0.306$  & $ 0.042$ & $0.008$ & $ 0.016$ \\   
 				Verizon Comm. Inc.            &  Comm. Services	&  $-0.122$  &  $0.217$  & $ 0.162$ & $0.007$ & $ 0.005$ \\   
 				Walgreens Boots All., Inc. &  Consumer Staples &  $-0.204$  &  $0.361$  & $ -0.190$ & $0.013$ & $ 0.017$ \\   
 				Walmart Inc.                           &  Consumer Staples	&  $ 0.071$  &  $0.237$  & $ 0.074$ & $0.006$ & $-0.006$ \\
 				\bottomrule
 			\end{tabular}}
 	\end{scriptsize}
    \end{center}
 \end{table}
 
Table \ref{tab:Description} reports, by company, the distribution mean and standard deviation of daily returns and ESG scores, as well as the correlation between returns and ESG scores. The dataset presents a wide cross-sectional variability in terms of the average return and risk, as well as in terms of the average ESG scores. ESG data show a smaller volatility compared to the returns despite a comparable cross-sectional dispersion of average returns and average ESG. This reflects the fact that the changes in the sustainability profile of a company happen at a much slower pace than the shifts seen in the monetary value of a stock. The correlation between the daily returns and ESG score is very close to zero for all the stocks. Figure \ref{fig:corr_matrix} displays the correlation matrix between the returns and ESG of each stock (the first 28 rows and columns are the returns, the latter 28 are the ESG values). We see that the returns show consistently positive and strong correlation with each other. In contrast, ESG-to-ESG correlations have mixed signs. ESG-to-return correlations are very close to zero. In terms of risk management, the lack of positive correlations among ESG and returns is beneficial to investors, as it allows diversification of risk between the two, supporting the motivation of the axiomatic framework.

\begin{figure}[!h]
	\centering
	\includegraphics[width=0.6\linewidth]{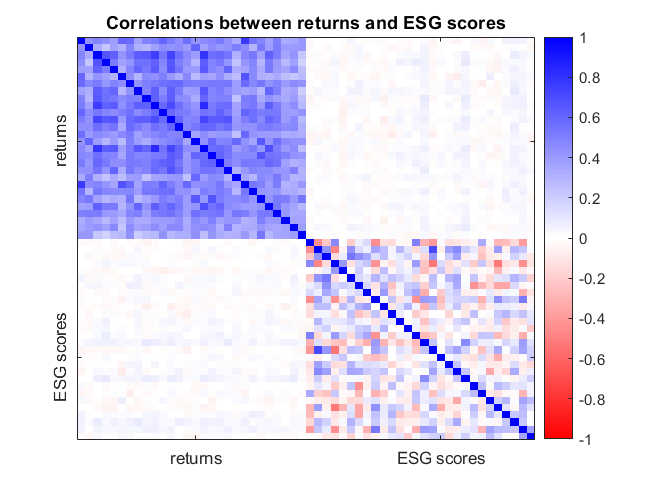}
	\caption{Correlations between returns and ESG of the companies in the sample.}
	\label{fig:corr_matrix}
\end{figure}

\begin{figure}[!h]
	\centering
	\includegraphics[width=1\linewidth]{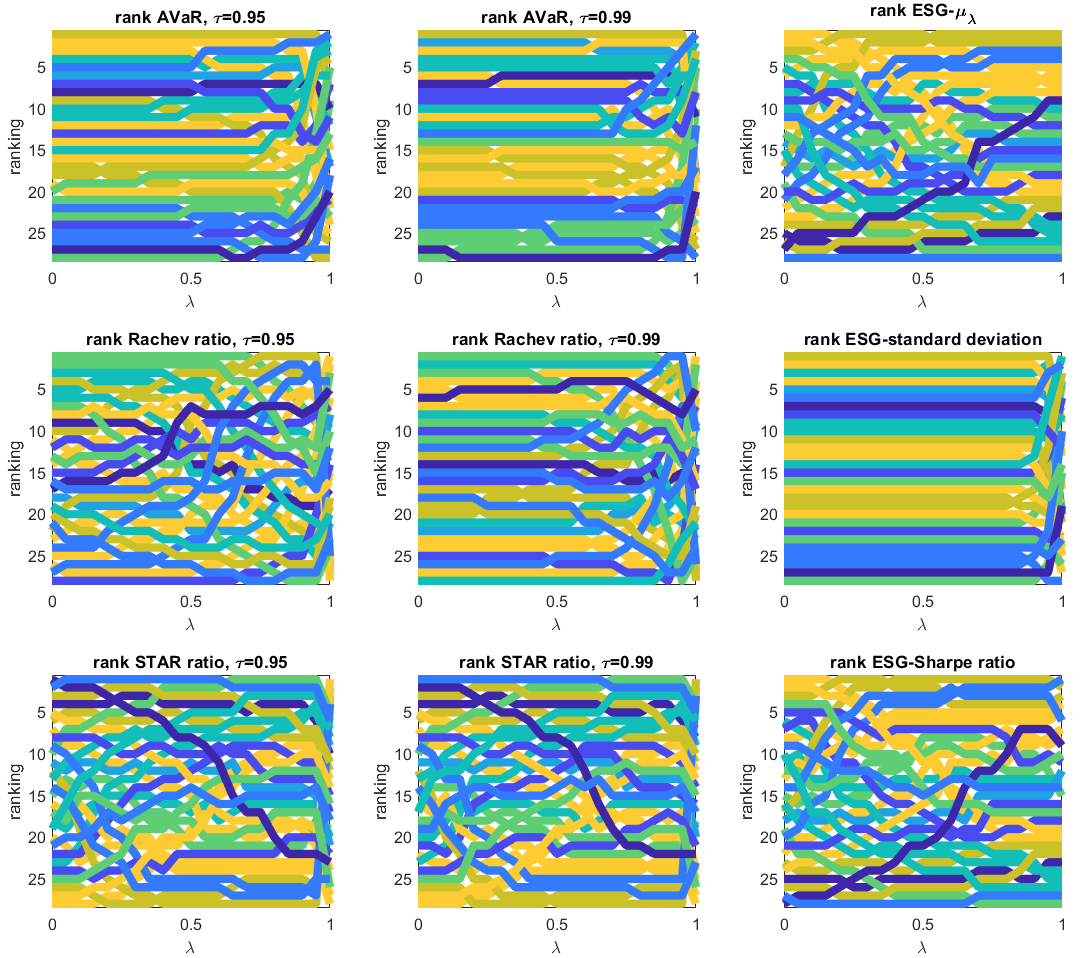}
	\includegraphics[width=0.9\linewidth]{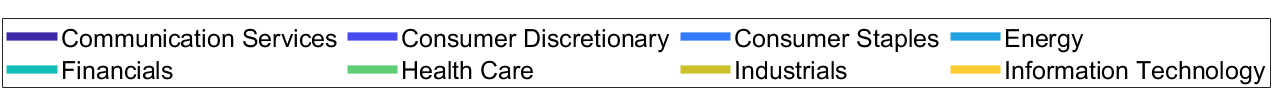}
	\caption{Visual representation of the ranking of the assets in the DJIA index based on $\esgd\avar_{\lambda,\tau}$ ($\tau$ equal to 0.95 and 0.99), $\esgd\mu_\lambda$, the ESG Rachev ratio ($\tau$ equal to 0.95 and 0.99), the standard deviation, $\esgd \starr_{\alpha,\lambda}$ (the STAR ratio with p = 2), and the Sharpe ratio. Values of $\lambda$ are arranged on the X axis from 0 (only the monetary component considered, left) to 1 (only the ESG component considered, right), and companies are color-coded by sector.}
	\label{fig:ranking}
\end{figure}

For each of the stocks, we compute a set of daily ESG risk measures and ESG reward--risk ratios, considering a range of values of $\lambda$. We then use the measures to rank the stocks in the dataset. In particular, we compute $\esgd \avar_{\lambda,\tau}$ (with $\tau$ equal to 0.95 and 0.99), the ESG standard deviation ($\esgd\sigma_\lambda$), the ESG-mean ($\esgd \mu_{\lambda}$), the ESG Rachev ratio (with $\tau$ equal to 0.95 and 0.99), the ESG STAR ratio (with $\tau$ equal to 0.95 and 0.99), and the ESG Sharpe ratio. The risk measures are computed using a historical simulation approach, employing all the observations in the period 2020-2023. The range of $\lambda$ considered is intentionally extreme and serves as a stress test for the framework. In practice, investors would typically choose smaller values of $\lambda$, resulting in portfolios that remain closer to conventional market portfolios.
Figure \ref{fig:ranking} displays the ordering of the assets: each company is color-coded according to the industrial sector in which it operates, and the companies in the upper part of the plot are those with the highest value for the indicator. On the horizontal axis, we vary the investor's $\lambda$ from zero to one. Tables \ref{tab:ranking1} and \ref{tab:ranking2} in Appendix \ref{sec:tables_rank} report the names of the companies in the top and bottom five positions for a selected sample of indicators with $\lambda \in \{ 0, 0.25, 0.5, 0.75, 1\}$.\footnote{For brevity we report only the measures for the 95th quantile.
	The complete rankings and the results for the 99th quantile are available upon request.}

We make the following observations: 

\begin{itemize}
	\item ESG-AVaR tends to give similar rankings for levels of $\lambda$ up to 0.5; then, the ranking changes significantly and converges to a significantly different ranking for $\lambda=1$ compared to $\lambda=0$. We can explain this behavior by taking into consideration that the ESG scores are characterized by a smaller variability over time; hence the tails of the distributions of returns is the most relevant determinant of the ESG-AVaR (and thus the ranking) even for relatively large values of $\lambda$. For irrealistically large values of $\lambda$ close to 1, the ESG component becomes dominant in the determination of the ranking. The results are similar for $\tau=0.95$ and $\tau=0.99$; the main difference is that for the highest value of $\tau$, the monetary component remains dominant for higher values of $\lambda$.
	\item The ranking according to the Rachev ratio is also stable for $\lambda$ smaller than  $0.5$. This is related to the fact that this ratio is computed as the ratio of the AVaRs for the top and bottom tails of the distribution.
	\item The ranking according to the STAR ratio shows significant changes when $\lambda$ changes. This can be attributed in large part to the fact that the numerator of the ratio ($\esgd\mu_\lambda$) by construction changes significantly with $\lambda$.
	\item The ranking according to the standard deviation is almost identical for all values of $\lambda$, except for values very close to 1. This is because, as stated previously, the variability of ESG scores is small, and the ESG standard deviation is almost exclusively driven by monetary returns. Unlike the AVaR, the standard deviation is not influenced by parallel shifts of the distribution; hence, increasing the value of $\lambda$ does not cause any changes in the ranking.
	\item $\esgd\mu_\lambda$, the ESG Sharpe ratio, and the ESG Sortino-Satchell ratio rankings change significantly for different values of $\lambda$. These changes are driven by the significant differences in the rankings of the average returns and average sustainability, which is also suggested by Panel C in Figure \ref{fig:plot_ESG_prices}.
\end{itemize}

Overall, the example shows the structural flexibility and interpretability of the proposed framework, which can be used to develop empirical analyses both to answer theoretical questions and to implement viable investment strategies suitable for ESG-oriented investors. Further empirical research is required for more extensive results, in particular for the modelization of the joint evolution of returns and ESG using stochastic process and the extension to more assets and geographical areas. Another natural extension is a benchmarking of the results against more ESG data providers, especially considering that the ESG scores from different rating agencies are often inconsistent \citep{billio2021inside}.

\section{Policy \& Practical Implications} \label{policy}

Modeling investment outcomes as a bivariate random vector $X_T = [ r_T, \; \esg_T ]'$, with 
$r_T$ denoting financial payoffs and $\esg_T$ the sustainability dimension, extends the standard univariate risk and reward measure framework in a structurally meaningful way. ESG-coherent risk measures embed sustainability considerations directly into the risk functional through a preference parameter $\lambda$, allowing ESG to enter portfolio choice in the same formal manner as other preference-driven dimensions. This provides a unified setting in which financial and ESG components are evaluated jointly, rather than sequentially or through external constraints.

For optimal asset allocation, this formulation yields well-defined optimization problems in which portfolio choice depends on the joint distribution of $X_t$, with only minimal modifications relative to standard univariate setups. ESG preferences are reflected through $\lambda$, inducing smooth adjustments in optimal portfolios reflecting the subjective trade-off between the monetary and the sustainability dimensions. Unlike exclusion or screening approaches, this mechanism preserves diversification properties and allows for a systematic analysis of sensitivity with respect to ESG preferences, dependence between $r_T$ and $\esg_T$, and joint tail behavior. From a theoretical perspective, it opens the way to studying ESG-aware allocation using familiar tools from coherent risk theory and convex optimization.

The framework is deliberately focused on ESG-oriented investing, in the sense that ESG enters the model as an argument of investor preferences rather than as a predictive signal for financial returns. This does not preclude the use of ESG information for improving financial performance, for example through forecasting, risk premia estimation, or alpha generation. Rather, the advantage of the proposed approach is to keep these two roles of ESG analytically distinct: ESG can be treated either as a component of the payoff vector $X_T$ reflecting non-financial objectives, or as an informational variable influencing the distribution of $r_T$. By separating preference-driven valuation from return-predictive use, the framework avoids conflating ethical objectives with performance claims and clarifies the economic interpretation of ESG integration, while remaining flexible enough to accommodate their joint use within advanced asset allocation strategies designed to address multiple, potentially heterogeneous, investment objectives.

The stochastic treatment of $\esg_T$ also has direct implications for ESG data provision. Meaningful ESG-coherent risk measurement requires time-series ESG data at frequencies compatible with financial modeling, creating incentives for data providers to move beyond static or purely cross-sectional scores. Higher-frequency and methodologically consistent ESG measures directly improve estimation of the joint law of $X_T = [ r_T, \; \esg_T ]'$, its dependence structure, and joint tail risk, rather than merely refining descriptive asset rankings.

In terms of transparency and reporting, ESG-coherent risk measures make the aggregation of financial and ESG outcomes explicit. The role of preferences is isolated in the parameter 
$\lambda$, which facilitates comparability across portfolios and products by separating differences due to investor objectives from differences due to underlying risk exposures. For institutional investors, this supports internal coherence between mandates, portfolio construction, and performance evaluation. For retail investors, it clarifies the trade-offs embedded in ESG-oriented products, reducing reliance on implicit or opaque screening rules.

Finally, the framework makes explicit several structural limitations that are often implicit in ESG integration. The choice of $\lambda$ is normative and cannot be identified from data alone, and the quality of ESG-coherent risk estimates is constrained by data availability and by the modeling of dependence between $r_T$ and $\esg_T$. Making these elements explicit shifts attention from the proliferation of ESG scores to the underlying statistical and preference-based structure of risk measurement for ESG-oriented investors, where further theoretical and empirical progress is most likely to be achieved.

\section{Conclusions} \label{sec:conclusions}

Individuals and institutions are increasingly aware of the non-monetary impact of their investments, and many are willing to structure their portfolios considering not only monetary risk and gains, but also the environmental, social, and governance implications. These ESG-oriented investors, who are rational under a richer set of objectives, challenge the traditional financial literature and require the development of new analytical tools to describe their behavior. Our work contributes to the literature by introducing ESG-coherent risk measures, a framework for the measurement of risk for investors with both monetary and ESG goals that provides an axiomatic definition in which risk is measured as a function of a bivariate random variable. The investor in our approach is still rational but follows a multi-dimensional evaluation that considers not only the monetary part but also sustainability. This framework extends the traditional coherent risk measures approach of \cite{artzner1999coherent}.

The measures we provide can be used in several contexts and, due to the introduction of a parameter $\lambda$, can be adapted to individuals with different relative preferences for the monetary and ESG components of risk. We also provide the dual representation for ESG-coherent risk measures, present several examples that generalize univariate risk measures, and introduce ESG-coherent reward measures and ESG reward--risk ratios.

We stress that the goal of the proposed approach is not to integrate ESG scores for improving monetary risk-adjusted performances, but to take into account the ethical preference of an investor for sustainable assets. This challenges the traditional assumptions of monetary profit-maximizing and risk-minimizing agents.

This paper is only an initial step in the development of a new financial theory that will be capable of describing the behavior of ESG-oriented investors. Future work will study optimal asset allocations, utility theory, and asset pricing for ESG-oriented investors. 

\section*{Acknowledgements}
This study was funded by the European Union - NextGenerationEU, in the framework of the GRINS - Growing Resilient, INclusive and Sustainable project (GRINS PE00000018 -- CUP F83C22001720001). The views and opinions expressed are solely those of the authors and do not necessarily reflect those of the European Union, nor can the European Union be held responsible for them.

\clearpage

\begin{appendices}

\section{ESG-coherent reward measures}\label{sec:reward}

Similarly to how we defined ESG-coherent risk measures, we define ESG-coherent reward measures. As discussed by \cite{rachev2008desirable}, in the univariate case, a reward measure can be defined as any functional defined on the space of the random variables of interest, that is, it should be isotonic with market preferences. Still, it is useful to formalize axiomatically the characteristics of reward measures. We now introduce axioms for multivariate reward measures inspired by the ones for the univariate case discussed in \cite{rachev2008desirable}. Considering a probability space $(\Omega, \mathcal{F}, P)$ and a set of bivariate random variables $\mathcal{X}_2$, we can define an ESG-adjusted reward measure $\bftheta_{\lambda}(X_T) : \mathcal{X}_2 \rightarrow \mathbb{R} \cup \{\pm \infty\}$, where $\lambda \in [0,1]$ is then ESG-coherent if the following axioms are satisfied:

\begin{itemize}[align = left]
	\item[(SUP-M+)] Super-additivity: if $X_{1,T} ,X_{2,T} \in \mathcal{X}_2$, then $\bftheta_\lambda(X_{1,T}+X_{2,T}) \geq \bftheta_\lambda(X_{1,T}) + \bftheta_\lambda(X_{2,T})$;
	\item [(PH-M+)]Positive Homogeneity: if $\beta\geq 0$ and $X_T \in \mathcal{X}_2$, then $\bftheta_\lambda(\beta X_T) = \beta \bftheta_\lambda(X_T)$;
	\item[(MO-M+)] Monotonicity: if $X_{1,T},X_{2,T} \in \mathcal{X}_2$ and ($r_{1,T}\leq r_{2,T} \wedge \esg_{1,T}\leq \esg_{2,T}$) a.s., then $\bftheta_\lambda(X_{1,T}) \leq \bftheta_\lambda(X_{2,T})$;
	\item [(LH-M+)] Lambda Homogeneity: if $a = [a_1, \;a_2]' \in \mathbb{R}^2$, then $\bftheta_\lambda(a) = (1-\lambda)a_1 + \lambda a_2$.
\end{itemize}
\vspace{0.5cm}
The translation invariance property follows directly from SUP-M+ and LH-M+, analogously to Proposition \ref{prop:TI}.

\begin{itemize}[align = left]
		\item [(TI-M+)] Translation Invariance: if $a = [a_1, \;a_2]'\in \mathbb{R}^2$, $\bftheta_\lambda(X_T+a)=\bftheta_\lambda(X_T) +  (1-\lambda)a_1 + \lambda a_2$.
\end{itemize}

\section{Proof of dual representation of ESG-coherent risk measures}\label{sec:proof_dual}

For convenience define $Z :=[Z_1, \; Z_2]' = -X_T$ where $X_T$ is the bivariate vector of monetary returns and ESG. Consider the space  $\Zc=\Lc_p(\Omega,\Fc,P;\Rb^2)$, $p\in [1,\infty]$.  When $p < \infty$, the dual space $\Zc^*$ is isomorphic to $\Lc_q(\Omega,\Fc,P;\Rb^2)$, where $q\in(1,\infty]$ is such that
$1/p+1/q=1$ with $q=\infty$ for $p=1$.


Consider the bilinear form $\langle\cdot,\cdot\rangle$ on the product $\Zc\times\Zc^*$, which is defined as follows. For $Z\in \Zc$ and $\zeta\in \Zc^*$ the value of the bilinear form is given by

\begin{equation}
\langle \zeta,Z \rangle  =\int_\Omega \left(\zeta_1(\omega)Z_1(\omega)+ \zeta_2(\omega)Z_2(\omega)\right) P(d\omega). \label{eq:bilinear1}
\end{equation}
The form provides the corresponding continuous linear functionals on $\Zc$ and $\Zc^*$ when equipped with appropriate topologies. For each fixed $\zeta\in \Zc^*$, the mapping $Z\mapsto \langle \zeta,Z \rangle $ is a continuous linear functional on $\Zc$, equipped with the norm topology. For $p\in (0,1),$ all continuous linear functional on $\Zc^*$ have this form. For $p=1$, we equip $Z^*$ with the weak$^*$ topology. 
 For $p = \infty$, the dual space $\Zc^*$ is formed by finitely-additive measures and it is inconvenient to work with.
 In this case, we pair $\Zc = \Lc_\infty(\Omega, \Fc,P,\Rb^2)$ with $\Lc_1(\Omega, \Fc,P,\Rb^2)$ and equip the latter space with its norm topology and the former with its weak$^*$ topology. We use the bilinear form \eqref{eq:bilinear1} with $Z \in \Lc_\infty(\Omega, \Fc,P,\Rb^2)$ and $\zeta \in \Lc_1(\Omega, \Fc,P,\Rb^2)$.
 
 Let $\bfrho_\lambda:\Zc\to\Rb\cup\{+\infty\}$ be a lower-semi-continuous functional with non-empty domain. In the case of $p = \infty$ we make the additional assumption that $\bfrho_\lambda$ is lower-semicontinuous with respect to its weak$^*$ topology.

The Fenchel conjugate function $\bfrho_\lambda^*:\Zc^*\to\overline{\Rb}$ of the risk measure $\bfrho_\lambda$ is defined as
\begin{equation*}
\bfrho_\lambda^*(\zeta)=\sup_{Z\in \Zc}\big\{\langle \zeta,Z\rangle-\bfrho_\lambda(Z)\big\},
\end{equation*}
and the conjugate of $\bfrho_\lambda^*$ (the bi-conjugate function) is defined as
\[
\bfrho_\lambda^{**}(Z)=\sup_{\zeta\in \Zc^*}\big\{\langle\zeta,Z\rangle-\bfrho_\lambda^*(\zeta)\big\}.
\]
$\Ac_{\bfrho_\lambda}$ denotes the domain of $\bfrho_\lambda^*.$ The Fenchel-Moreau Theorem, which is valid in paired spaces, implies that $\bfrho_\lambda^{**}(Z) =\bfrho_\lambda(Z)$ whenever $\bfrho_\lambda$ is proper, convex, and lower-semicontinuous.

\underline{Claim:} The MO-M property holds iff for all $\zeta\in\Ac_{\bfrho_\lambda}$, $\zeta\geq 0$ a.s.

Assume that the opposite is true. This means that there exists a set $\Delta\in\Fc$ with $P(\Delta)>0$ such that for $\omega\in\Delta$, we have $\zeta_i(\omega)< 0$ for $i=1$ or $i=2$. We define $\bar{Z}_i = \one_{\Delta \cap {\zeta_i<0}}$, where $\one_B$ is the indicator function of the event $B$. Take any $Z$ with support in $\Delta$ such that $\bfrho_\lambda(Z)$ is finite and and define  $Z_t:=Z-t\bar{Z}$. Then for $t\geq 0$, we have  that  $Z_t\leq Z$ componentwise, and
	$\bfrho_\lambda(Z_t)\leq \bfrho_\lambda(Z)$ by monotonicity. Consequently,
\[
\bfrho_\lambda^*(\zeta)\geq \sup_{t\in \Rb_+}\big\{\langle \zeta,Z_t\rangle-\bfrho_\lambda(Z_t)\big\}
	\geq \sup_{t\in \Rb_+}\big\{\langle
	\zeta,Z\rangle-t\langle\zeta,\bar{Z}\rangle  -\bfrho_\lambda(Z)\big\}.
\]
On the right-hand side, $\langle\zeta,\bar{Z}\rangle <0$ on $\Delta$ and zero otherwise, while the other terms under the supremum are fixed. Hence, the supremum is infinite and $\zeta\not\in\Ac_{\bfrho_\lambda}.$

Conversely, suppose that every $\zeta\in \Ac_{\bfrho_\lambda}$ is nonnegative. Then for every $\zeta\in\Ac_{\bfrho_\lambda}$ and $Z\geq Z'$ componentwise,
we have  
\[
\langle \zeta,Z'\rangle = \int_\Omega \left(\zeta_1(\omega)Z_1'(\omega)+ \zeta_2(\omega)Z_2'(\omega) \right)P(d\omega) \leq \int_\Omega \left(\zeta_1(\omega)Z_1(\omega)+ \zeta_2(\omega)Z_2(\omega) \right) P(d\omega)=  \langle \zeta,Z\rangle. 
\]
Consequently 
\[
\bfrho_\lambda(Z)=\sup_{\zeta\in \Zc^*}\big\{\langle
\zeta,Z\rangle-\bfrho_\lambda^*(\zeta)\big\}\geq \sup_{\zeta\in \Zc^*}\big\{\langle
\zeta,Z'\rangle-\bfrho_\lambda^*(\zeta)\big\} = \bfrho_\lambda(Z').
\]

\underline{Claim:} The PH-M property holds iff $\bfrho_\lambda$ is the support function of $\Ac_{\bfrho_\lambda}$.

Suppose that $\bfrho_\lambda(tZ)=t\bfrho_\lambda(Z)$ for all $Z\in\Zc$.
For all $t>0$ and for all $Z\in\Zc$
\[
\bfrho_\lambda^*(\zeta) = \sup_{Z\in \Zc}\big\{\langle \zeta,Z\rangle-\bfrho_\lambda(Z)\big\}
\geq \langle
\zeta,tZ\rangle -\bfrho_\lambda(tZ)
\]
Thus for all $t>0$ 
\[
\bfrho_\lambda^*(\zeta) = \sup_{Z\in \Zc}\big\{\langle \zeta,Z\rangle-\bfrho_\lambda(Z)\big\}
\geq \sup_{Z\in \Zc}\big\{\langle
\zeta,tZ\rangle -t\bfrho_\lambda(Z)\big\}=t\bfrho_\lambda^*(\zeta).
\]
Hence, if $\bfrho_\lambda^*(\zeta)$ is finite, then $\bfrho_\lambda^*(\zeta) = 0$ as claimed.
Furthermore,  
\[
	\bfrho_\lambda (0) = \sup_{\zeta\in \Zc^*}\big\{\langle \zeta,0\rangle-\bfrho_\lambda^*(\zeta)\big\} = 0.
\]
For the converse, if $\bfrho_\lambda(Z) = \sup_{\zeta\in\Ac_{\bfrho_\lambda}}\langle \zeta,Z\rangle$, then $\bfrho_\lambda$ is positively homogeneous as a support function of a convex set. 
 Hence when the PH-M property holds, the conjugate function is the indicator function of convex analysis of the set $\Ac_{\bfrho_\lambda}.$

\underline{Claim:}  Assume $\bfrho_\lambda$ is a proper, convex, and positively homogeneous risk functional. Then the risk measure is additive for any constant vectors $a, b \in \Rb$ (i.e. $\bfrho_\lambda(a+b) = \bfrho_\lambda(a) + \bfrho_\lambda(b)$) iff $\mu_\zeta = \int_\Omega \zeta(\omega)P(d\omega) = \mu$ and $\langle\bfone,\mu\rangle = \bfrho_\lambda(\bfone)$, for all $\zeta \in \Ac$, where $\bfone = [1, \;1]'$. Furthermore, for all $Z \in \Xc_2$ and $a \in \Rb^2$, 
$$\bfrho_{\lambda}(Z+a) = \bfrho_{\lambda}(Z) + \bfrho_{\lambda}(a). $$\\
Let us denote $\int_\Omega\zeta(\omega)P(d\omega) = \mu_\zeta$. If $\zeta \in \Ac_{\bfrho_\lambda}$ and SUB-M and PH-M hold, then for all $a \in \Rb^2$
\begin{equation}
	\bfrho_{\lambda}(a) = - \sup_{\zeta \in \Zc^*}\langle a, \mu_\zeta\rangle \label{eq:dual_constant}
\end{equation}
and $\bfrho_{\lambda}(\bfzero) = 0$, where $\bfzero = [0,\;0]'$. If $\mu = \mu_\zeta = \int_\Omega\zeta(\omega)P(d\omega)$ for all $\zeta \in \Ac$, then equation \ref{eq:dual_constant} implies that the risk measure is additive for any constant vectors. Assume now that the measure is additive for constant vectors. Hence,
$$0 = \bfrho_{\lambda}(\bfzero) =  \bfrho_{\lambda}(a - a) =   \bfrho_{\lambda}(a) +  \bfrho_{\lambda}(-a).$$

Consequently $\bfrho_{\lambda}$ is linear on constants and $\langle a, \mu_\zeta\rangle = \langle a,\mu\rangle$ with $\mu = \mu_\zeta$ for all $a \in \Rb^2$ and  for all $\zeta \in \Ac_{\bfrho_\lambda}$. Indeed,
$$ \sup_{\zeta \in \Zc^*}\langle a, \mu_\zeta\rangle = \sup_{\zeta \in \Zc^*}\langle -a, \mu_\zeta\rangle = 0$$
or equivalently
$$ \sup_{\zeta \in \Zc^*}\langle a, \mu_\zeta\rangle = \inf_{\zeta \in \Zc^*}\langle a, \mu_\zeta\rangle.$$
Hence $\langle a, \mu_\zeta\rangle$ is the same for all $\zeta \in \Ac_{\bfrho_\lambda}$. Since $a \in \Rb^2$ is arbitrary, we conclude that $\mu = \mu_\zeta$. For any $\zeta \in \Ac_{\bfrho_\lambda}$ we have
$$\int_\Omega\langle \bfone,\zeta(\omega)\rangle P(d\omega) = \langle\bfone,\mu_\zeta\rangle = \sup_{\zeta \in \Zc^*}\langle \bfone, \mu_\zeta\rangle = -\bfrho_\lambda(\bfone).$$
Let $Z \in \Zc$, $X_T = -Z$, and $a \in \Rb$.
\begin{align}
	\bfrho_\lambda(X_T+a) & = \sup_{\zeta \in \Zc^*}\left\{\int_\Omega\langle Z - a,\zeta(\omega)\rangle P(d\omega) \right\}  \notag\\
	  & = \sup_{\zeta \in \Zc^*}\left\{\int_\Omega\langle Z,\zeta(\omega)\rangle P(d\omega) - \langle a,\mu \rangle \right\}  = \bfrho_\lambda(X_T) + \bfrho_\lambda(a).\notag
\end{align}

\underline{Claim:} If additionally LH-M holds, then $\mu=  \big[(1-\lambda),\lambda\big]'$ and, hence, for all $\zeta\in\Ac_{\bfrho_\lambda}$, we have
\[
\int_\Omega \zeta_1(\omega)+ \zeta_2(\omega) P(d\omega) =1.
\]

In summary, when axioms SUB-M through LH-M are satisfied, then the dual representation of the risk measure is 
\begin{align}\label{eq:dual_appendix}
\bfrho_\lambda (X_T)& =\sup_{\zeta\in \Ac_{\bfrho_\lambda}} \left\{- \int_\Omega \left(\zeta_1(\omega) r_T (\omega)+ \zeta_2(\omega) \esg_T (\omega) \right) P(d\omega) \right\} \\
\bfrho_\lambda (X_T)& =-\inf_{\zeta\in \Ac_{\bfrho_\lambda}} \left\{ \int_\Omega \left( \zeta_1(\omega) r_T (\omega)+ \zeta_2(\omega) \esg_T (\omega) \right) P(d\omega) \right\} \notag
\end{align}
where $\Ac_{\bfrho_\lambda}$ contains non-negative functions $(\zeta_1(\omega),\zeta_2(\omega))$ on $\Rb^2$ whose expected value is \mbox{$\big[(1-\lambda), \quad \lambda\big]'$}.
Furthermore, $\Ac_{\bfrho_\lambda}$ is equal to the convex subdifferential of $\bfrho_\lambda ([0, 0]')$. Note that in \eqref{eq:dual_appendix} we adjusted the signs since we express the risk measure in terms of $X_T$ and not $Z=-X_T$.

\section{Proofs of dual representation for $\esgavar$ and $\esgavar^l$}\label{sec:proof_avar_dual}

Let $X_T=[r_T, \; \esg_T]' \in \Xc_2=\mathcal{L}_1(\Omega,\mathcal{F},P;\Rb^2)$ be a bivariate random variable associated with an asset, and  for convenience, define $Z :=[Z_1, \; Z_2]' = -X_T$ (i.e., the corresponding vector with inverted signs).

\textbf{Dual representation for $\esgavar$}

Here, we prove that $\esgavar(X_T)$ has the following dual representation:
\begin{gather*}
\esgavar(X_T) = \sup_{[\zeta_1 \; \zeta_2]' \in \Ac_{\esgavar}} \Eb[Z_1\zeta_1+Z_2 \zeta_2],
\quad \text{with}\\
\Ac_{\esgavar} = \left\{[\zeta_1, \; \zeta_2]'\in \mathcal{L}_\infty(\Omega,\mathcal{F},P;\Rb^2):\; [\zeta_1, \; \zeta_2]'=\xi[1-\lambda, \; \lambda]'; \; 0\leq \xi\leq \frac{1}{1 -\tau}\; \text{a.s. } \Eb[\xi]= 1\right\}.
\end{gather*}
We assume $\Xc_2=\mathcal{L}_1(\Omega,\mathcal{F},P;\Rb^2)$, which entails that the paired space is $\mathcal{L}_\infty(\Omega,\mathcal{F},P;\Rb^2)$. 
\begin{align*}
\esgavar(X_T) & = \min_{\beta\in \Rb}\Big\{ \frac{1}{1-\tau} \Eb\big[ \big(\beta-((1-\lambda)r+\lambda\esg)\big)_+\big] -\beta\Big\},\\
& = \min_{\beta\in \Rb}\Big\{\beta + \frac{1}{1-\tau} \Eb\big[ \big((1-\lambda)Z_1+\lambda Z_2-\beta\big)_+\big]\Big\},
\quad \tau\in (0,1).
\end{align*}
Using the rules of subdifferential calculus and Strassen's theorem \citep{strassen1965existence}, we get
\[
\partial  \Eb\big[ \big((1-\lambda)Z_1+\lambda Z_2-\beta\big)_+\big] = 
[1-\lambda, \; \lambda]'\xi,\quad \text{where } \xi(\omega) = \begin{cases}
1 & \text{ if } (1-\lambda)Z_1+\lambda Z_2 > \beta\\
0  & \text{ if } (1-\lambda)Z_1+\lambda Z_2 < \beta \\
[0,1] & \text{ if } (1-\lambda)Z_1+\lambda Z_2 = \beta.
\end{cases}
\]
Note that $\xi\in \mathcal{L}_\infty(\Omega,\mathcal{F},P)$ and $\xi\geq 0$. We define $\zeta\in \mathcal{L}_\infty(\Omega,\mathcal{F},P;\Rb^2)$ by setting $\zeta = [1-\lambda, \; \lambda]'\xi$ for any measurable selection $\xi\in \partial  \E\big[ \big((1-\lambda)r+\lambda\esg)\big)_+\big].$ Let $\tilde{\mathcal{A}}$ be the set containing all such elements $\zeta$. Evidently, we have $\bfzero \leq\zeta\leq [1-\lambda, \; \lambda]'$ a.s. for all. 
Using the subgradient inequality at $\bar{Z}= (\beta,\beta)$,  we obtain for all $Z$ 
\[
\E\big[ \big(((1-\lambda)Z_1+\lambda Z_2)-\beta\big)_+\big] \geq \langle \zeta, Z-\beta(1,1)\rangle = \langle \zeta, Z\rangle - \beta\E[\xi].
\]
On the other hand, for any $\zeta=[1-\lambda, \; \lambda]'\xi$
\begin{multline*}
\langle \zeta,Z\rangle =\langle [1-\lambda, \; \lambda]' \xi, Z-\beta(1,1)\rangle + \beta\langle  [1-\lambda, \; \lambda]' \xi, (1,1)\rangle \\
= \langle \xi,((1-\lambda)Z_1+\lambda Z_2 -\beta)\rangle +\beta\E[\xi]
  \leq  \E\big[ \big((1-\lambda)Z_1+\lambda Z_2 -\beta)_+\big] +\beta\E[\xi].
\end{multline*}
Hence, we can represent 
\[
\E\big[ \big((1-\lambda)Z_1+\lambda Z_2 -\beta)_+\big]  = \max_{\zeta\tilde{\Ac}} \Big(\langle \zeta,Z\rangle  - \beta\Eb[\xi]\Big). 
\]
Now, we can express the risk measure as follows:
\begin{align*}
\bfrho(Z) &= \min_{\beta\in \Rb}\bigg\{ \beta + \frac{1}{1-\tau} \max_{\zeta\in \tilde{\Ac}} \Big(\langle \zeta,Z\rangle  - \beta\Eb[\xi]\Big) \bigg\} \\ &= \max_{\zeta\in \tilde{\Ac}} \inf_{\beta\in \Rb}\bigg\{ \beta + \frac{1}{1-\tau} \langle \zeta,Z\rangle - \frac{\beta}{1-\tau}\Eb[\xi] \bigg\} = \max_{\zeta\in \tilde{\Ac}} \inf_{\beta\in \Rb}\bigg\{ \beta (1-\frac{1}{1-\tau}\Eb[\xi] ) + \frac{1}{1-\tau} \langle \zeta,Z\rangle\bigg\},
\end{align*}
where $\bfrho(Z) = \esgavar(X_T)$. The exchange of the ``$\min$'' and ``$\max$'' operations is possible, because the function in braces is bilinear in $(\beta,\zeta)$, and the set $\tilde{\Ac}$ is compact. 
The inner minimization with respect to $\beta$ yields $-\infty$, unless $\Eb[\xi]=1-\tau$. Consequently,
\[
\bfrho(Z) = \max_{{\zeta= \xi  [1-\lambda,\lambda]'\in \Ac}\atop {\Eb[\xi]=1-\tau}} \frac{1}{1-\tau} \langle Z,\zeta\rangle.
\]
Setting $\zeta' = \zeta/(1-\tau)$,  we obtain the support set
\[
\partial \bfrho(0) = \bigg\{\zeta \in \Lc_\infty(\Omega,\Fc,P;\Rb^2): \zeta = \frac{\xi}{1-\tau} [1-\lambda, \lambda]'; \xi\ge 0,\ \Eb[\xi]= 1-\tau, \ \|\xi\|_\infty \le 1\bigg\}.
\]

\textbf{Dual representation for $\esgavar^l$}

We now prove that the dual representation of $\esgavar^l(X_T)$ is defined as follows:
\begin{equation}
\esgavar^l(X_T) = \sup_{[\zeta_1, \; \zeta_2]' \in \Ac_{\esgavar}} \Eb[Z_1\zeta_1 + Z_2 \zeta_2],
\end{equation}

with $Z :=[Z_1, \; Z_2]' = -X_T$ and 
\begin{equation}
\Ac_{\esgavar^l} = \left\{[\zeta_1, \; \zeta_2]'\in \Lc_\infty(\Omega,\Fc,P;\Rb^2): \Eb[\zeta_1]=1-\lambda;\Eb[\zeta_2]=\lambda;\zeta_1,\zeta_2\geq 0;\zeta_1\leq\frac{1-\lambda}{1-\tau};\zeta_2\leq\frac{\lambda}{1-\tau}\right\}.
\end{equation}

We use the known representation of Average Value at Risk for scalar random variables (cf. \citealp[Example 6.19 eq. 6.76]{shapiro2021lectures}). 
Denote the dual set in that representation by $\Ac^\prime$, i.e.,  
\begin{equation*}
\Ac^\prime = \left\{\xi\in\mathcal{L}_\infty(\Omega,\Fc,P):\; \Eb[\xi]=1;\; 0\leq \xi \leq\frac{1}{1-\tau}\; \text{a.s.}\right\},\notag\\
\end{equation*}
Hence
\begin{align}
\esgavar^l(X) 
 &=(1-\lambda)\sup_{\xi \in \Ac^\prime} \Eb[ \xi, Z_1 ] + \lambda\sup_{\xi \in \Ac^\prime} \Eb[ \xi, Z_2 ]\notag\\
 &=\sup_{\xi \in \Ac^\prime} \Eb[ (1-\lambda)\xi, Z_1 ] + \sup_{\xi \in \Ac^\prime} \Eb[ \lambda\xi, Z_2 ]\notag\\
 &=\sup_{[\xi_1,\xi_2]' \in \Ac^\prime\times\Ac^\prime}\Big( \Eb[ (1-\lambda)\xi_1, Z_1 ] +\Eb[ \lambda\xi_2, Z_2 ]\Big).
\end{align}
Now, we define the set $\Ac_{\esgavar^l}= (1-\lambda)\Ac^\prime\times\lambda\Ac^\prime$ and continue the last chain of equations as follows:
\begin{align*}
\esgavar^l(X_T) 
 &=\sup_{[\xi_1,\xi_2]' \in \Ac^\prime\times\Ac^\prime} \Eb[(1-\lambda)\xi, Z_1] +\Eb[ \lambda\xi, Z_2 ]= \sup_{\zeta\in \Ac_{\esgavar^l} }\langle \zeta, Z \rangle. 
\end{align*}
This concludes the proof.

\pagebreak
\begin{landscape}
\section{Company rankings based on ESG risk, reward, and RRR measures}\label{sec:tables_rank}
\vspace{-0.3in}
\begin{table}[!h]		
	\begin{tiny}
	\caption{Top and bottom ranking according to several ESG risk and reward measures ($\esgd \avar_{\lambda, 0.95\%}$,
			ESG-standard deviation, ESG-mean) for selected values of $\lambda$.}
	\label{tab:ranking1}
		\begin{tabular}{c LLLLL}
			\toprule
			\multicolumn{6}{c}{ranking by $\esgd \avar_{\lambda, 0.95\%}$}\\
			rank & $\lambda = 0$                             & $\lambda = 0.25$                          & $\lambda = 0.5$                           & $\lambda = 0.75$             & $\lambda = 1$                     \\
			\midrule
			1  & Boeing Company                  &  Boeing Company                  & Boeing Company                 &  Boeing Company                  & Walgreens Boots Alliance, Inc.              \\
			2  & Intel Corporation               &  Intel Corporation               & Intel Corporation              &  Intel Corporation               & Travelers Companies, Inc.                   \\
			3  & Salesforce, Inc.                &  Salesforce, Inc.                & Salesforce, Inc.               &  Walgreens Boots Alliance, Inc.  & Goldman Sachs Group, Inc.                   \\
			4  & American Express Company        &  American Express Company        & American Express Company       &  American Express Company        & 3M Company                                  \\
			5  & Walgreens Boots Alliance, Inc.  &  Walgreens Boots Alliance, Inc.  & Walgreens Boots Alliance, Inc. &  Salesforce, Inc.                & Johnson \& Johnson                           \\
			\multicolumn{6}{c}{\strut}\\
			24 & McDonald's Corporation          &  McDonald's Corporation          & McDonald's Corporation         &  Coca-Cola Company               & Salesforce, Inc.                            \\
			25 & Walmart Inc.                    &  Walmart Inc.                    & Walmart Inc.                   &  McDonald's Corporation          & IBM Corporation \\
			26 & Procter \& Gamble Company       &  Procter \& Gamble Company       & Procter \& Gamble Company      &  Johnson \& Johnson              & Caterpillar Inc.                            \\
			27 & Verizon Communications Inc.     &  Verizon Communications Inc.     & Verizon Communications Inc.    &  Procter \& Gamble Company       & Coca-Cola Company                           \\
			28 & Johnson \& Johnson              &  Johnson \& Johnson              & Johnson \& Johnson             &  Verizon Communications Inc.     & Honeywell International Inc.                \\
			\multicolumn{6}{c}{\strut}\\	\toprule
			\multicolumn{6}{c}{ranking by ESG-standard deviation}\\
			rank & $\lambda = 0$                             & $\lambda = 0.25$                          & $\lambda = 0.5$                           & $\lambda = 0.75$             & $\lambda = 1$                     \\
			\midrule
			1  & Boeing Company              &  Boeing Company              &  Boeing Company               & Boeing Company              &  Travelers Companies, Inc.                   \\
			2  & Intel Corporation           &  Intel Corporation           &  Intel Corporation            & Intel Corporation           &  Walgreens Boots Alliance, Inc.              \\
			3  & American Express Company    &  American Express Company    &  American Express Company     & American Express Company    &  3M Company                                  \\
			4  & Salesforce, Inc.            &  Salesforce, Inc.            &  Salesforce, Inc.             & Salesforce, Inc.            &  Amgen Inc.                                  \\
			5  & Chevron Corporation         &  Chevron Corporation         &  Chevron Corporation          & Chevron Corporation         &  American Express Company                    \\
			\multicolumn{6}{c}{\strut}\\
			24 & Walmart Inc.                &  Walmart Inc.                &  Walmart Inc.                 & Walmart Inc.                &  Microsoft Corporation                       \\
			25 & Coca-Cola Company           &  Coca-Cola Company           &  Coca-Cola Company            & Coca-Cola Company           &  Apple Inc.                                  \\
			26 & Procter \& Gamble Company   &  Procter \& Gamble Company   &  Procter \& Gamble Company    & Procter \& Gamble Company   &  Merck \& Co., Inc.                           \\
			27 & Verizon Communications Inc. &  Verizon Communications Inc. &  Verizon Communications Inc.  & Verizon Communications Inc. &  IBM Corporation \\
			28 & Johnson \& Johnson          &  Johnson \& Johnson          &  Johnson \& Johnson           & Johnson \& Johnson          &  Coca-Cola Company                           \\
			\multicolumn{6}{c}{\strut}\\		\toprule
			\multicolumn{6}{c}{ranking by ESG-mean}\\
			rank & $\lambda = 0$                             & $\lambda = 0.25$                          & $\lambda = 0.5$                           & $\lambda = 0.75$             & $\lambda = 1$                     \\
			\midrule
			1  & Apple Inc.                       & Caterpillar Inc.                 & Caterpillar Inc.                 & Honeywell International Inc.     & Honeywell International Inc.   \\
			2  & Microsoft Corporation            & Apple Inc.                       & Honeywell International Inc.     & Caterpillar Inc.                 & Caterpillar Inc.               \\
			3  & Caterpillar Inc.                 & Microsoft Corporation            & Salesforce, Inc.                 & Coca-Cola Company                & Coca-Cola Company              \\
			4  & UnitedHealth Group Inc.  & Salesforce, Inc.                 & Procter \& Gamble Company         & Salesforce, Inc.                 & Procter \& Gamble Company       \\
			5  & Goldman Sachs Group, Inc.        & Honeywell International Inc.     & Coca-Cola Company                & Procter \& Gamble Company         & Salesforce, Inc.               \\
			\multicolumn{6}{c}{\strut}\\
			24 & Boeing Company                   & Boeing Company                   & Boeing Company                   & 3M Company                       & JPMorgan Chase \& Co.           \\
			25 & Walt Disney Company              & Verizon Communications Inc.      & Johnson \& Johnson                & Walt Disney Company              & Goldman Sachs Group, Inc.      \\
			26 & 3M Company                       & 3M Company                       & 3M Company                       & Travelers Companies, Inc.        & Travelers Companies, Inc.      \\
			27 & Verizon Communications Inc.      & Walt Disney Company              & Walt Disney Company              & Johnson \& Johnson                & Johnson \& Johnson              \\
			28 & Walgreens Boots Alliance, Inc.   & Walgreens Boots Alliance, Inc.   & Walgreens Boots Alliance, Inc.   & Walgreens Boots Alliance, Inc.   & Walgreens Boots Alliance, Inc. \\
			\bottomrule
		\end{tabular}
	\end{tiny}
\end{table}
\end{landscape}

\pagebreak

\begin{landscape}
\begin{table}[!h]		
	\begin{tiny}
	\caption{Top and bottom ranking according to several ESG-RRRs (ESG Sharpe ratio, $\esgd \starr_{\lambda, 0.95\%}$,
			$\esgd \rr_{\lambda, 0.95\%}$) for selected values of $\lambda$.}
	\label{tab:ranking2}	
		\begin{tabular}{c LLLLL}
			\toprule
			\multicolumn{6}{c}{ranking by ESG Sharpe ratio}\\
			rank & $\lambda = 0$                             & $\lambda = 0.25$                          & $\lambda = 0.5$                           & $\lambda = 0.75$             & $\lambda = 1$                     \\ 
			\midrule
			1  & Apple Inc.                         & Caterpillar Inc.                  & Honeywell International Inc.                  & Honeywell International Inc.                & Coca-Cola Company                           \\
			2  & Microsoft Corporation              & Apple Inc.                        & Caterpillar Inc.                              & Coca-Cola Company                           & Honeywell International Inc.                \\
			3  & Caterpillar Inc.                   & Microsoft Corporation             & Procter \& Gamble Company                      & Procter \& Gamble Company                    & Caterpillar Inc.                            \\
			4  & UnitedHealth Group Inc.    & Honeywell International Inc.      & Coca-Cola Company                             & Caterpillar Inc.                          IBM  & IBM Corporation \\
			5  & McDonald's Corporation             & McDonald's Corporation            & IBM Corporation   & IBM Corporation & Salesforce, Inc.                            \\
			\multicolumn{6}{c}{\strut}\\
			24 & Intel Corporation                  & Johnson \& Johnson                 & Boeing Company                                & Walt Disney Company                         & JPMorgan Chase \& Co.                        \\
			25 & Walt Disney Company                & Walt Disney Company               & Walt Disney Company                           & 3M Company                                  & Goldman Sachs Group, Inc.                   \\
			26 & 3M Company                         & Verizon Communications Inc.       & 3M Company                                    & Travelers Companies, Inc.                   & Travelers Companies, Inc.                   \\
			27 & Walgreens Boots Alliance, Inc.     & 3M Company                        & Johnson \& Johnson                             & Walgreens Boots Alliance, Inc.              & Walgreens Boots Alliance, Inc.              \\
			28 & Verizon Communications Inc.        & Walgreens Boots Alliance, Inc.    & Walgreens Boots Alliance, Inc.                & Johnson \& Johnson                           & Johnson \& Johnson                           \\
			\multicolumn{6}{c}{\strut}\\	\toprule
			\multicolumn{6}{c}{ranking by $\esgd \starr_{\lambda, 0.95\%}$}\\
			rank & $\lambda = 0$                             & $\lambda = 0.25$                          & $\lambda = 0.5$                           & $\lambda = 0.75$             & $\lambda = 1$                     \\
			\midrule
			1  & Verizon Communications Inc.       & Walgreens Boots Alliance, Inc.   & Walgreens Boots Alliance, Inc.               & Johnson \& Johnson                             & IBM Corporation \\
			2  & Walgreens Boots Alliance, Inc.    & 3M Company                       & Johnson \& Johnson                            & Walgreens Boots Alliance, Inc.                & Caterpillar Inc.                            \\
			3  & 3M Company                        & Verizon Communications Inc.      & 3M Company                                   & Travelers Companies, Inc.                     & Honeywell International Inc.                \\
			4  & Walt Disney Company               & Walt Disney Company              & Walt Disney Company                          & 3M Company                                    & Coca-Cola Company                           \\
			5  & Intel Corporation                 & Johnson \& Johnson                & Boeing Company                               & Walt Disney Company                           & Johnson \& Johnson                           \\
			\multicolumn{6}{c}{\strut}\\
			24 & McDonald's Corporation            & McDonald's Corporation           & IBM Corporation  & IBM Corporation   & Cisco Systems, Inc.                         \\
			25 & UnitedHealth Group Inc.   & Honeywell International Inc.     & Coca-Cola Company                            & Caterpillar Inc.                              & Procter \& Gamble Company                    \\
			26 & Caterpillar Inc.                  & Microsoft Corporation            & Procter \& Gamble Company                     & Coca-Cola Company                             & Merck \& Co., Inc.                           \\
			27 & Microsoft Corporation             & Apple Inc.                       & Caterpillar Inc.                             & Procter \& Gamble Company                      & Intel Corporation                           \\
			28 & Apple Inc.                        & Caterpillar Inc.                 & Honeywell International Inc.                 & Honeywell International Inc.                  & Salesforce, Inc.                            \\
			\multicolumn{6}{c}{\strut}\\	\toprule
			\multicolumn{6}{c}{ranking by $\esgd \rr_{\lambda, 0.95\%}$}\\
			rank & $\lambda = 0$                             & $\lambda = 0.25$                          & $\lambda = 0.5$                           & $\lambda = 0.75$             & $\lambda = 1$                     \\
			\midrule
			1  & Amgen Inc.                                    & Amgen Inc.                                    & Amgen Inc.                        & Amgen Inc.                       & Salesforce, Inc.                            \\
			2  & UnitedHealth Group Inc.               & UnitedHealth Group Inc.               & UnitedHealth Group Inc.   & Honeywell International Inc.     & Intel Corporation                           \\
			3  & JPMorgan Chase \& Co.                          & JPMorgan Chase \& Co.                          & Caterpillar Inc.                  & Caterpillar Inc.                 & Merck \& Co., Inc.                           \\
			4  & Boeing Company                                & American Express Company                      & JPMorgan Chase \& Co.              & UnitedHealth Group Inc.  & Procter \& Gamble Company                    \\
			5  & American Express Company                      & Boeing Company                                & American Express Company          & Procter \& Gamble Company         & Verizon Communications Inc.                 \\
			\multicolumn{6}{c}{\strut}\\
			24 & Procter \& Gamble Company                      & Walgreens Boots Alliance, Inc.                & Coca-Cola Company                 & Cisco Systems, Inc.              & Johnson \& Johnson                           \\
			25 & IBM Corporation   & IBM Corporation   & 3M Company                        & Johnson \& Johnson                & Coca-Cola Company                           \\
			26 & Cisco Systems, Inc.                           & Coca-Cola Company                             & Walgreens Boots Alliance, Inc.    & 3M Company                       & Honeywell International Inc.                \\
			27 & Coca-Cola Company                             & Cisco Systems, Inc.                           & Cisco Systems, Inc.               & Home Depot, Inc.                 & Caterpillar Inc.                            \\
			28 & Home Depot, Inc.                              & Home Depot, Inc.                              & Home Depot, Inc.                  & Walgreens Boots Alliance, Inc.   & IBM Corporation \\
			\bottomrule
		\end{tabular}
	\end{tiny}
\end{table}
\end{landscape}

\end{appendices}

\pagebreak

\bibliography{bibliography.bib}

@article{ruszczynski2006optimization,
  title={Optimization of convex risk functions},
  author={Ruszczy{\'n}ski, Andrzej and Shapiro, Alexander},
  journal={Mathematics of Operations Research},
  volume={31},
  number={3},
  pages={433--452},
  year={2006},
  publisher={INFORMS}
}

@article{rockafellar2013fundamental,
	title={The fundamental risk quadrangle in risk management, optimization and statistical estimation},
	author={Rockafellar, R Tyrrell and Uryasev, Stan},
	journal={Surveys in Operations Research and Management Science},
	volume={18},
	number={1-2},
	pages={33--53},
	year={2013},
	publisher={Elsevier}
}

@article{artzner1999coherent,
	title={Coherent measures of risk},
	author={Artzner, Philippe and Delbaen, Freddy and Eber, Jean-Marc and Heath, David},
	journal={Mathematical Finance},
	volume={9},
	number={3},
	pages={203--228},
	year={1999},
	publisher={Wiley Online Library}
}

@article{follmer2002convex,
	title={Convex measures of risk and trading constraints},
	author={F{\"o}llmer, Hans and Schied, Alexander},
	journal={Finance and Stochastics},
	volume={6},
	number={4},
	pages={429--447},
	year={2002},
	publisher={Springer}
}

@article{ekeland2012comonotonic,
	title={Comonotonic measures of multivariate risks},
	author={Ekeland, Ivar and Galichon, Alfred and Henry, Marc},
	journal={Mathematical Finance},
	volume={22},
	number={1},
	pages={109--132},
	year={2012},
	publisher={Wiley Online Library}
}

@article{rachev2008desirable,
	title={Desirable properties of an ideal risk measure in portfolio theory},
	author={Rachev, Svetlozar and Ortobelli, Sergio and Stoyanov, Stoyan and Fabozzi, Frank J and Biglova, Almira},
	journal={International Journal of Theoretical and Applied Finance},
	volume={11},
	number={01},
	pages={19--54},
	year={2008},
	publisher={World Scientific}
}

@article{cheridito2013reward,
	title={Reward--risk ratios},
	author={Cheridito, Patrick and Kromer, Eduard},
	journal={Journal of Investment Strategies},
	volume={3},
	number={1},
	pages={3--18},
	year={2013}
}

@article{ruschendorf2006law,
	title={Law invariant convex risk measures for portfolio vectors},
	author={R{\"u}schendorf, Ludger},
	journal={Statistics \& Risk Modeling},
	volume={24},
	number={1},
	pages={97--108},
	year={2006},
	publisher={Oldenbourg Wissenschaftsverlag GmbH}
}

@article{jouini2004vector,
	title={Vector-valued coherent risk measures},
	author={Jouini, Elyes and Meddeb, Moncef and Touzi, Nizar},
	journal={Finance and Stochastics},
	volume={8},
	number={4},
	pages={531--552},
	year={2004},
	publisher={Springer}
}

@article{hamel2009duality,
	title={A duality theory for set-valued functions {I}: Fenchel conjugation theory},
	author={Hamel, Andreas H},
	journal={Set-Valued and Variational Analysis},
	volume={17},
	number={2},
	pages={153--182},
	year={2009},
	publisher={Springer}
}

@article{wei2014coherent,
	title={Coherent and convex risk measures for portfolios with applications},
	author={Wei, Linxiao and Hu, Yijun},
	journal={Statistics \& Probability Letters},
	volume={90},
	pages={114--120},
	year={2014},
	publisher={Elsevier}
}

@article{chen2020multivariate,
	title={Multivariate coherent risk measures induced by multivariate convex risk measures},
	author={Chen, Yanhong and Hu, Yijun},
	journal={Positivity},
	volume={24},
	number={3},
	pages={711--727},
	year={2020},
	publisher={Springer}
}

@article{hamel2013set,
	title={Set-valued average value at risk and its computation},
	author={Hamel, Andreas H and Rudloff, Birgit and Yankova, Mihaela},
	journal={Mathematics and Financial Economics},
	volume={7},
	number={2},
	pages={229--246},
	year={2013},
	publisher={Springer}
}

@article{biglova2004different,
	title={Different approaches to risk estimation in portfolio theory},
	author={Biglova, Almira and Ortobelli, Sergio and Rachev, Svetlozar T and Stoyanov, Stoyan},
	journal={The Journal of Portfolio Management},
	volume={31},
	number={1},
	pages={103--112},
	year={2004},
	publisher={Institutional Investor Journals Umbrella}
}

@book{sortino2001managing,
	title={Managing Downside Risk in Financial Markets},
	author={Sortino, Frank A and Satchell, Stephen},
	year={2001},
	publisher={Butterworth-Heinemann},
        address={Oxford}
}

@article{utz2015tri,
	title={Tri-criterion modeling for constructing more-sustainable mutual funds},
	author={Utz, Sebastian and Wimmer, Maximilian and Steuer, Ralph E},
	journal={European Journal of Operational Research},
	volume={246},
	number={1},
	pages={331--338},
	year={2015},
	publisher={Elsevier}
}

@article{ang2018dual,
	title={On the dual representation of coherent risk measures},
	author={Ang, Marcus and Sun, Jie and Yao, Qiang},
	journal={Annals of Operations Research},
	volume={262},
	number={1},
	pages={29--46},
	year={2018},
	publisher={Springer}
}

@article{markowitz1952portfolio,
	title={Portfolio selection},
	author={Markowitz, Harry},
	journal={The Journal of Finance},
	volume={7},
	number={1},
	pages={77--91},
	year={1952},
	publisher={Wiley Online Library}
}

@Article{cesarone2022does,
	AUTHOR = {Cesarone, Francesco and Martino, Manuel Luis and Carleo, Alessandra},
	TITLE = {Does {ESG} Impact Really Enhance Portfolio Profitability?},
	JOURNAL = {Sustainability},
	VOLUME = {14},
	YEAR = {2022},
	NUMBER = {4},
	ARTICLE-NUMBER = {2050},
	URL = {https://www.mdpi.com/2071-1050/14/4/2050},
	ISSN = {2071-1050}
}

@article{rockafellar2000optimization,
	title={Optimization of conditional value-at-risk},
	author={Rockafellar, R Tyrrell and Uryasev, Stanislav and others},
	journal={Journal of Risk},
	volume={2},
	pages={21--42},
	year={2000},
	publisher={Citeseer}
}

@article{rockafellar2002conditional,
  title={Conditional value-at-risk for general loss distributions},
  author={Rockafellar, R Tyrrell and Uryasev, Stanislav},
  journal={Journal of Banking \& Finance},
  volume={26},
  number={7},
  pages={1443--1471},
  year={2002},
  publisher={Elsevier}
}

@article{ogryczak2002dual,
  title={Dual stochastic dominance and related mean-risk models},
  author={Ogryczak, WLodzimierz and Ruszczy{\'n}ski, Andrzej},
  journal={SIAM Journal on Optimization},
  volume={13},
  number={1},
  pages={60--78},
  year={2002},
  publisher={SIAM}
}

@article{giacometti2021tail,
	title={Tail risks in large portfolio selection: {P}enalized quantile and expectile minimum deviation models},
	author={Giacometti, Rosella and Torri, Gabriele and Paterlini, Sandra},
	journal={Quantitative Finance},
	volume={21},
	number={2},
	pages={243--261},
	year={2021},
	publisher={Taylor \& Francis}
}

@article{pedersen2021responsible,
	title={Responsible investing: The {ESG}-efficient frontier},
	author={Pedersen, Lasse Heje and Fitzgibbons, Shaun and Pomorski, Lukasz},
	journal={Journal of Financial Economics},
	volume={142},
	number={2},
	pages={572--597},
	year={2021},
	publisher={Elsevier}
}

@article{black1989universal,
	title={Universal hedging: Optimizing currency risk and reward in international equity portfolios},
	author={Black, Fischer},
	journal={Financial Analysts Journal},
	volume={45},
	number={4},
	pages={16--22},
	year={1989},
	publisher={Taylor \& Francis}
}

@article{hornuf2024performance,
	title={The performance of socially responsible investments: A meta-analysis},
	author={Hornuf, Lars and Y{\"u}ksel, G{\"u}l},
	journal={European Financial Management},
	volume={30},
	number={2},
	pages={1012--1061},
	year={2024},
	publisher={Wiley Online Library}
}

@article{revelli2015financial,
	title={Financial performance of socially responsible investing ({SRI}): {W}hat have we learned? {A} meta-analysis},
	author={Revelli, Christophe and Viviani, Jean-Laurent},
	journal={Business Ethics: A European Review},
	volume={24},
	number={2},
	pages={158--185},
	year={2015},
	publisher={Wiley Online Library}
}

@book{rachev2011risk,
	title={Advanced Stochastic Models, Risk Assessment, and Portfolio Optimization. The Ideal Risk, Uncertainty, and Performance},
	author={Rachev, Svetlozar T and Stoyanov, Stoyan V and Fabozzi, Frank J},
	year={2011},
	publisher={Wiley},
        address={Hoboken}
}

@article{keating2002universal,
	title={A universal performance measure},
	author={Keating, Con and Shadwick, William F},
	journal={Journal of Performance Measurement},
	volume={6},
	number={3},
	pages={59--84},
	year={2002}
}

@article{farinelli2008sharpe,
	title={Sharpe thinking in asset ranking with one-sided measures},
	author={Farinelli, Simone and Tibiletti, Luisa},
	journal={European Journal of Operational Research},
	volume={185},
	number={3},
	pages={1542--1547},
	year={2008},
	publisher={Elsevier}
}

@article{sandberg2009heterogeneity,
	title={The heterogeneity of socially responsible investment},
	author={Sandberg, Joakim and Juravle, Carmen and Hedesstr{\"o}m, Ted Martin and Hamilton, Ian},
	journal={Journal of Business Ethics},
	volume={87},
	number={4},
	pages={519--533},
	year={2009},
	publisher={Springer}
}

@article{widyawati2020systematic,
	title={A systematic literature review of socially responsible investment and environmental social governance metrics},
	author={Widyawati, Luluk},
	journal={Business Strategy and the Environment},
	volume={29},
	number={2},
	pages={619--637},
	year={2020},
	publisher={Wiley Online Library}
}

@article{amel2018and,
	title={Why and how investors use {ESG} information: Evidence from a global survey},
	author={Amel-Zadeh, Amir and Serafeim, George},
	journal={Financial Analysts Journal},
	volume={74},
	number={3},
	pages={87--103},
	year={2018},
	publisher={Taylor \& Francis}
}

@article{billio2021inside,
	title={Inside the {ESG} Ratings: ({D}is)agreement and performance},
	author={Billio, Monica and Costola, Michele and Hristova, Iva and Latino, Carmelo and Pelizzon, Loriana},
	journal={Corporate Social Responsibility and Environmental Management},
	volume={28},
	number={5},
	pages={1426--1445},
	year={2021},
	publisher={Wiley Online Library}
}

@article{berg2022aggregate,
	title={Aggregate confusion: The divergence of {ESG} ratings},
	author={Berg, Florian and Koelbel, Julian F and Rigobon, Roberto},
	journal={Review of Finance},
	volume={26},
	number={6},
	pages={1315--1344},
	year={2022},
	publisher={Oxford University Press}
}

@article{schanzenbach2020reconciling,
	title={Reconciling fiduciary duty and social conscience: The law and economics of {ESG} investing by a trustee},
	author={Schanzenbach, Max M and Sitkoff, Robert H},
	journal={Stan. L. Rev.},
	volume={72},
	pages={381},
	year={2020},
	publisher={HeinOnline}
}

@book{shapiro2021lectures,
  title={Lectures on Stochastic Programming: Modeling and Theory},
  author={Shapiro, Alexander and Dentcheva, Darinka and Ruszczy{\'n}ski, Andrzej},
  year={2021},
  publisher={SIAM},
  address={Philadelphia}
}

@article{strassen1965existence,
  title={The existence of probability measures with given marginals},
  author={Strassen, Volker},
  journal={The Annals of Mathematical Statistics},
  volume={36},
  number={2},
  pages={423--439},
  year={1965},
  publisher={Institute of Mathematical Statistics}
}

@article{heinkel2001effect,
	title={The effect of green investment on corporate behavior},
	author={Heinkel, Robert and Kraus, Alan and Zechner, Josef},
	journal={Journal of financial and quantitative analysis},
	volume={36},
	number={4},
	pages={431--449},
	year={2001},
	publisher={Cambridge University Press}
}

@article{pastor2021sustainable,
	title={Sustainable investing in equilibrium},
	author={P{\'a}stor, L'ubo{\v{s}} and Stambaugh, Robert F and Taylor, Lucian A},
	journal={Journal of financial economics},
	volume={142},
	number={2},
	pages={550--571},
	year={2021},
	publisher={Elsevier}
}

@article{gasser2017markowitz,
	title={Markowitz revisited: Social portfolio engineering},
	author={Gasser, Stephan M and Rammerstorfer, Margarethe and Weinmayer, Karl},
	journal={European Journal of Operational Research},
	volume={258},
	number={3},
	pages={1181--1190},
	year={2017},
	publisher={Elsevier}
}

@article{feinstein2013time,
	title={Time consistency of dynamic risk measures in markets with transaction costs},
	author={Feinstein, Zachary and Rudloff, Birgit},
	journal={Quantitative Finance},
	volume={13},
	number={9},
	pages={1473--1489},
	year={2013},
	publisher={Taylor \& Francis}
}

@book{dentcheva2024risk,
	title={Risk-Averse Optimization and Control: Theory and Methods},
	author={Dentcheva, Darinka and Ruszczy{\'n}ski, Andrzej},
	year={2024},
	publisher={Springer Nature},
        address   ={Cham}
}

@article{liu2017distributionally,
  title={Distributionally robust reward-risk ratio optimization with moment constraints},
  author={Liu, Yongchao and Meskarian, Rudabeh and Xu, Huifu},
  journal={SIAM Journal on Optimization},
  volume={27},
  number={2},
  pages={957--985},
  year={2017},
  publisher={SIAM}
}

@article{lauria2026mean,
  title={Mean-{CVaR} portfolio optimization under {ESG} disagreement},
  author={Lauria, Davide and Bonomelli, Marco and Torri, Gabriele and Giacometti, Rosella},
  journal={Computational Management Science},
  volume={23},
  number={1},
  pages={5},
  year={2026},
  publisher={Springer}
}

@article{friede2015esg,
  title={{ESG} and financial performance: aggregated evidence from more than 2000 empirical studies},
  author={Friede, Gunnar and Busch, Timo and Bassen, Alexander},
  journal={Journal of sustainable finance \& investment},
  volume={5},
  number={4},
  pages={210--233},
  year={2015},
  publisher={Taylor \& Francis}
}
\bibliographystyle{apalike}

\end{document}